\documentclass[12pt]{article}
\usepackage{amssymb,amsmath,amsthm,amsfonts,mathrsfs}
\usepackage{enumerate}

\theoremstyle{definition}

\usepackage{pgf,tikz}
\usetikzlibrary{intersections}
\usepackage{graphicx}
\usepackage{float}
\usepackage[colorlinks=true, allcolors=blue]{hyperref}
\usepackage[all]{hypcap}
\usepackage[capitalize,nameinlink]{cleveref}
\usepackage{titlesec}
\titlelabel{\thetitle.\quad}
\newtheorem{theo}{Theorem}

\theoremstyle{definition}
\newtheorem{case}{Case}[theo]

\newtheorem{definition}{Definition}
\begin{document}
\title{Circular-arc $H$-graphs: Ordering Characterizations
and Forbidden Patterns}
\author{ Indrajit Paul, Ashok Kumar Das\\Department of Pure Mathematics, University of Calcutta\\
Email Address -  
paulindrajit199822@gmail.com \&\\ ashokdas.cu@gmail.com}
\maketitle
\begin{abstract}
We introduce the class of circular-arc 
$H$-graphs, which generalizes circular-arc graphs, particularly circular-arc bigraphs. We investigate two types of ordering-based characterizations of circular-arc $r$-graphs. Finally, we provide forbidden obstructions for circular-arc $r$-graphs in terms of specific vertex orderings.
                                                            
\end{abstract}
\section{Introduction}
A graph $G = (V, E)$ is called a \emph{circular-arc graph} if it is the intersection graph of 
circular arcs on a host circle. A bipartite graph (or simply, a \emph{bigraph}) 
$B = (X, Y, E)$ is a \emph{circular-arc bigraph} if there exists a family 
$\mathcal{A} = \{ A_v : v \in X \cup Y \}$ of circular arcs such that 
$uv \in E$ if and only if $A_u \cap A_v \neq \varnothing$, where $u \in X$ and $v \in Y$.  

The problem of characterizing circular-arc graphs was first studied by Klee~\cite{klee}. 
Circular-arc graphs and their subclasses, such as \emph{proper circular-arc graphs} 
(where no arc is properly contained in another in the representation) and 
\emph{Helly circular-arc graphs} (where the family of arcs satisfies the Helly property), 
have been extensively investigated by Tucker and others~\cite{Gavril, tucker1, tucker,tucker2}. 
More recently, Francis, Hell, and Stacho~\cite{fhs} presented an obstruction 
characterization and a certifying recognition algorithm for circular-arc graphs.  

In contrast, the bipartite version of circular-arc graphs, namely 
\emph{circular-arc bigraphs}, remains relatively less explored. 
Sen et al.~\cite{sdw} introduced circular-arc di/bigraphs and provided several 
characterizations of circular-arc bigraphs. Proper circular-arc bigraphs were studied 
by Das and Chakraborty~\cite{dc} and by Safe~\cite{safe}. 
Most of these characterizations rely on the adjacency matrix. 
In a recent work, Paul and Das~\cite{paul_das} characterized circular-arc bigraphs using 
vertex orderings and also provided forbidden-pattern characterizations with respect to 
specific vertex orderings.

An important direction, however, remains largely unexplored: the \emph{generalization of 
circular-arc bigraphs to graphs with more than two partite sets}. In this paper, we study 
and characterize these generalized classes in several ways. Recently, Müller and Rafiey~\cite{muller_rafiey} 
extended the concept of interval bigraphs by introducing \emph{interval $H$-graphs}. 
Motivated by their work, we introduce the analogous class of \emph{circular-arc $H$-graphs}, 
which generalizes circular-arc bigraphs.  

Formally, for a fixed graph $H$ with vertices $h_1, h_2, \ldots, h_r$, we say that an input 
graph $G$ with a vertex partition $V_1, V_2, \ldots, V_r$ is a \emph{circular-arc $H$-graph} 
if each vertex $v \in G$ can be represented by a circular arc $A_v$ on a host circle such that 
for $u \in V_i$ and $v \in V_j$, the vertices $u$ and $v$ are adjacent in $G$ if and only if 
$h_i h_j \in E(H)$ and the arcs $A_u$ and $A_v$ intersect. In particular, $G$ is called a 
\emph{circular-arc $r$-graph} when $H$ is the complete graph on $r$ vertices, and a 
\emph{circular-arc bigraph} when $r = 2$.
\section{Main Result}
In this section, we introduce several types of vertex orderings for $r$-partite graphs and examine their role in characterizing circular-arc $r$-graphs. We present characterizations of circular-arc $r$-graphs that are based on these orderings. In addition, we establish a characterization in terms of forbidden patterns, which may be regarded as one of the most intriguing characterizations of circular-arc $r$-graphs to date. 
We begin by defining the notion of a \textit{generalized total-circular ordering} of the vertices of a $r$-partite graph.

\begin{definition}[\textbf{Generalized total-circular ordering}]\label{d1}
Consider a $r$-partite graph $B=(X_1,X_2,...,X_r,E)$ of order $n$. Order the vertices of $B$ from $1$ to $n$ and arrange them on an $n$-hour clock, such that the $i^{\text{th}}$ vertex is on the $i^{\text{th}}$ hour marker. Assume that the vertex set $X=\bigcup_{i=1}^{r} X_i $ satisfyies the following conditions:
\begin{enumerate}[(a)]

    \item $x_ix_j\in E$ with $i>j$ (where $x_i\in X_{\alpha}$, $x_j\in X_{\beta}$, $\alpha\neq\beta$), then:
    \begin{itemize}
        \item either $x_ix_k\in E$, for all  $x_k\notin X_{\alpha}$, where $k\in \{ j+1, j+2,...,i-2,i-1\}$
        \item or, $x_lx_j\in E$, for all $x_l\notin X_{\beta}$, where $l\in \{i+1,i+2,...,n,1,2,...,j-1\}$,
       
    \end{itemize}
     \item $x_ix_j\in E$ with $i<j$ (where $x_i\in X_{\alpha}$, $x_j\in X_{\beta}$, $\alpha\neq\beta$), then:
    \begin{itemize}
        \item either $x_kx_j\in E$, for all  $x_k\notin X_{\beta}$, where $k\in \{ i+1, i+2,...,j-2,j-1\}$
        \item or, $x_ix_l\in E$, for all $x_l\notin X_{\alpha}$, where $l\in \{j+1,j+2,...,n,1,2,...,i-1\}$.
    \end{itemize}
\end{enumerate}
Then the vertex set $X=\bigcup_{i=1}^{r} X_i $ of $B$ is said to have a \textit{generalized total-circular ordering}.\\
 Using this generalized total-circular ordering of a $r$-partite graph, we will characterize circular-arc $r$-graphs in the following theorem:
\end{definition}
\begin{theo}\label{t3}
    \textit{An $r$-partite graph $B=(X_1,X_2,...X_r,E)$ is a circular-arc $r$-graph if and only if the vertex set $X=\bigcup_{i=1}^{r} X_i $ of $B$ admits a \textit{generalized total-circular ordering}}.
   
\end{theo}
\begin{proof}
    Necessity: Let $B=(X_1,X_2,...,X_r,E)$ be a circular-arc $r$-graph. Then there exist a circular arc $A_v$ corresponding to every vertex $v$ of the set $X=\bigcup_{i=1}^{r} X_i $. Such that $uv\in E$ if and only if $A_u\cap A_v\neq\emptyset$, where $u$ and $v$ belongs to different partite sets. Without loss of generality we consider that all the arcs having distinct end points. Now order the vertices of $B$ from $1$ to $n$ according to increasing order of clockwise end points (where $n$ is the order of the $r$-partite graph $B$). Let $v_1, v_2, v_3, ..., v_n$ be such an ordering. 
    \par Let $v_i$ be adjacent to $v_j$, where $v_i\in X_{\alpha}$, $v_j\in X_{\beta}$, and $\alpha\neq \beta$. Then, we have the following cases:
    \begin{case}\label{c1}
        ($i>j$)
    \end{case}
    \begin{itemize}

   \item Either, clockwise end point of $A_{v_i}$ (arc corresponding to $v_i$) lies within $A_{v_j}$ ( arc corresponding to $v_j$), in this case $A_{v_k}\cap A_{v_j}\neq \phi$, for all $k\in \{i+1, i+2,...,n,1,...,j-1\}$.\\
   Therefore $v_kv_j\in E$ for all $v_k\notin X_{\beta}$, where $k\in\{i+1,i+2,...,n,1,...,j-1\}$.
   \item Or, clockwise end point of $A_{v_j}$ lies within $A_{v_i}$, in this case $A_{v_l}\cap A_{v_i}\neq\emptyset$, for all $l\in\{ j+1, j+2,...,i-1\}$.\\
   Therefore $v_iv_l\in E$ for all $v_l\notin X_{\alpha}$, where $l\in\{ j+1, j+2,...,i-1\}$.
    \end{itemize}
    \begin{case}\label{c2}
        ($i<j$)
    \end{case}
    \begin{itemize}
    \item Either, clockwise end point of $A_{v_i}$ (arc corresponding to $v_i$) lies within $A_{v_j}$ ( arc corresponding to $v_j$), in this case $A_{v_k}\cap A_{v_j}\neq \emptyset$, for all $k\in \{i+1, i+2,...,j-1\}$.\\
   Therefore $v_kv_j\in E$ for all $v_k\notin X_{\beta}$, where $k\in \{i+1, i+2,...,j-1\}$.
   \item Or, clockwise end point of $A_{v_j}$ lies within $A_{v_i}$, in this case $A_{v_l}\cap A_{v_i}\neq\emptyset$, for all $l\in\{ j+1, j+2,...,n,1,...,i-1\}$.\\
   Therefore $v_iv_l\in E$ for all $v_l\notin X_{\alpha}$, where $l\in\{ j+1, j+2,...,n,1,...,i-1\}$.
    \end{itemize}
    Hence, the ordering $v_1,v_2,...,v_n$ of  vertices of the bigraph $B$ is a generalized total-circular ordering.
    \par\noindent{}Sufficiency: Let $B=(X_1,X_2,...,X_r,E)$ be an $r$-partite graph where the vertices are ordered as $v_1,v_2,...,v_n$, which is a generalized total-circular ordering.\\
    Now, we will construct a circular arc for each vertex of  the $r$-partite graph $B$. Let $k$ be the $k^{\text{th}}$ hour marker on an n-hour clock.\\
    If $v_i\in X_{\alpha}$, then draw a closed arc $A_i$ anticlockwise from $i$ to $m_i$, where $v_{m_i}$ is the last consecutive vertex from the set $X\setminus X_{\alpha}$ in the anticlockwise sequence $v_{i-1},v_{i-2},...,v_i$ that is adjacent to $v_i$ (i.e. $A_i=[m_i, i]$).\\
     If $v_j\in X_{\beta}$, then draw a closed arc $A_j$ anticlockwise from $j$ to $m_j$, where $v_{m_j}$ is the last consecutive vertex from the set  $X\setminus X_{\beta}$ in the anticlockwise sequence $v_{j-1},v_{j-2},...,v_j$ that is adjacent to $v_j$ (i.e. $A_j=[m_j, j]$).
     \par If $v_i$ is adjacent to $v_j$, for some $v_i\in X_{\alpha}$ and $v_j\in X_{\beta}$, then we have the following cases:
     \begin{case}\label{c3}
         ($i>j$)
     \end{case}
     Then, by \cref{d1} of generalized total-circular ordering:
     \begin{itemize}
          \item either $v_iv_k\in E$, for all $v_k\notin X_{\alpha}$, where ($k\in \{ j+1, j+2,...,i-2,i-1\}$), then $A_i$ contains $j$ and therefore $A_i\cap A_j\neq\emptyset$.
          
        \item or, $v_lv_j\in E$, for all $v_l\notin X_{\beta}$,where ($l\in \{i+1,i+2,...,n,1,2,...,j-1\}$), then $A_j$ contains $i$ and therefore $A_i\cap A_j\neq\phi$.
     \end{itemize}
     \begin{case}\label{c4}
         ($i<j$)
     \end{case}
     \begin{itemize}
        \item either $v_kv_j\in E$, for all $v_k\notin X_{\beta}$ ($k\in \{ i+1, i+2,...,j-2,j-1\}$), then $A_j$ contains $i$ and therefore $A_i\cap A_j\neq\emptyset$.
        \item or, $v_iv_l\in E$, for all $v_l\notin X_{\alpha}$ ($l\in \{j+1,j+2,...,n,1,2,...,i-1\}$), then $A_i$ contains $j$ and therefore $A_i\cap A_j\neq\emptyset$.
    \end{itemize}
    Therefore in any case, $v_i$ is adjacent to $v_j$ implies $A_i\cap A_j\neq\phi$.\\
    Again, let $A_i\cap A_j\neq\phi$, where $A_i$ is the circular arc corresponding to the vertex $v_i$ and $A_j$ is the circular arc corresponding to the vertex $v_j$. Then by the construction of the circular arcs it is clear that the vertex $v_i$ is adjacent to $v_j$.\\
    Thus $v_iv_j\in E$ if and only if $A_i\cap A_j\neq\phi$.
    Therefore $B=(X_1,X_2,...,X_r,E)$ is a circular-arc $r$-graph. 
    
\end{proof}

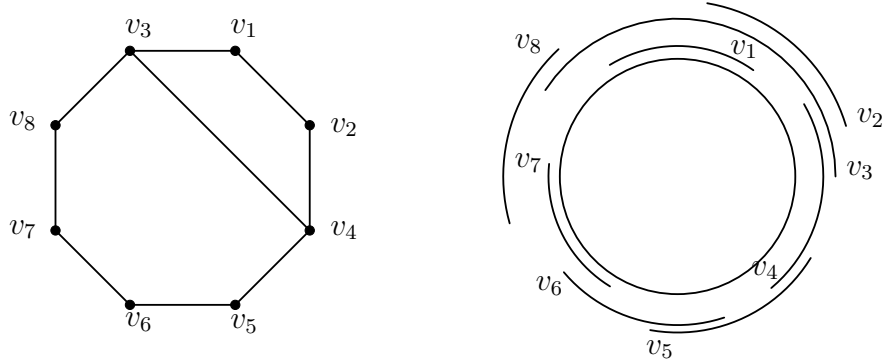
\begin{figure}[H]{\label{f1}}
    \centering
    \begin{tikzpicture}[line cap=round,line join=round,x=1.0cm,y=1.0cm,scale=.7]
\clip(-6.,1.) rectangle (11.,9.);
\fill[line width=.7pt,color=black,fill=white,fill opacity=0.10000000149011612] (-3.4,2.56) -- (-1.4,2.56) -- (0.014213562373094568,3.9742135623730945) -- (0.014213562373095012,5.974213562373094) -- (-1.4,7.388427124746189) -- (-3.4,7.3884271247461895) -- (-4.814213562373094,5.9742135623730945) -- (-4.814213562373094,3.9742135623730954) -- cycle;
\draw [line width=.7pt] (7.,5.) circle (2.23606797749979cm);
\draw [shift={(7.,5.)},line width=.7pt]  plot[domain=0.:2.562953072305964,variable=\t]({1.*3.*cos(\t r)+0.*3.*sin(\t r)},{0.*3.*cos(\t r)+1.*3.*sin(\t r)});
\draw [shift={(7.,5.)},line width=.7pt]  plot[domain=0.9514840476604486:2.1169628902286846,variable=\t]({1.*2.480725700274014*cos(\t r)+0.*2.480725700274014*sin(\t r)},{0.*2.480725700274014*cos(\t r)+1.*2.480725700274014*sin(\t r)});
\draw [shift={(7.,5.)},line width=.7pt]  plot[domain=-0.872357789965374:0.5201245875154695,variable=\t]({1.*2.768176294963888*cos(\t r)+0.*2.768176294963888*sin(\t r)},{0.*2.768176294963888*cos(\t r)+1.*2.768176294963888*sin(\t r)});
\draw [shift={(7.,5.)},line width=.7pt]  plot[domain=0.29145679447786715:1.4016951007935197,variable=\t]({1.*3.3408980828513757*cos(\t r)+0.*3.3408980828513757*sin(\t r)},{0.*3.3408980828513757*cos(\t r)+1.*3.3408980828513757*sin(\t r)});
\draw [shift={(7.,5.)},line width=.7pt]  plot[domain=4.536154282648014:5.73369576085061,variable=\t]({1.*2.96593998590666*cos(\t r)+0.*2.96593998590666*sin(\t r)},{0.*2.96593998590666*cos(\t r)+1.*2.96593998590666*sin(\t r)});
\draw [shift={(7.,5.)},line width=.7pt]  plot[domain=3.8417705860151123:5.0318822022723415,variable=\t]({1.*2.824535360019414*cos(\t r)+0.*2.824535360019414*sin(\t r)},{0.*2.824535360019414*cos(\t r)+1.*2.824535360019414*sin(\t r)});
\draw [shift={(7.,5.)},line width=.7pt]  plot[domain=3.0435473762959537:4.1623146077095,variable=\t]({1.*2.451774867315513*cos(\t r)+0.*2.451774867315513*sin(\t r)},{0.*2.451774867315513*cos(\t r)+1.*2.451774867315513*sin(\t r)});
\draw [shift={(7.,5.)},line width=.7pt]  plot[domain=2.3220197665761497:3.4141405272160372,variable=\t]({1.*3.3111931384321265*cos(\t r)+0.*3.3111931384321265*sin(\t r)},{0.*3.3111931384321265*cos(\t r)+1.*3.3111931384321265*sin(\t r)});
\draw [line width=.7pt,color=black] (-3.4,2.56)-- (-1.4,2.56);
\draw [line width=.7pt,color=black] (-1.4,2.56)-- (0.014213562373094568,3.9742135623730945);
\draw [line width=.7pt,color=black] (0.014213562373094568,3.9742135623730945)-- (0.014213562373095012,5.974213562373094);
\draw [line width=.7pt,color=black] (0.014213562373095012,5.974213562373094)-- (-1.4,7.388427124746189);
\draw [line width=.7pt,color=black] (-1.4,7.388427124746189)-- (-3.4,7.3884271247461895);
\draw [line width=.7pt,color=black] (-3.4,7.3884271247461895)-- (-4.814213562373094,5.9742135623730945);
\draw [line width=.7pt,color=black] (-4.814213562373094,5.9742135623730945)-- (-4.814213562373094,3.9742135623730954);
\draw [line width=.7pt,color=black] (-4.814213562373094,3.9742135623730954)-- (-3.4,2.56);
\draw [line width=.7pt] (-3.4,7.3884271247461895)-- (0.014213562373094568,3.9742135623730945);
\draw (-3.7,8.2) node[anchor=north west] {$v_3$};
\draw (-1.7,8.2) node[anchor=north west] {$v_1$};
\draw (0.2,6.32) node[anchor=north west] {$v_2$};
\draw (0.2,4.36) node[anchor=north west] {$v_4$};
\draw (-1.7,2.58) node[anchor=north west] {$v_5$};
\draw (-3.7,2.62) node[anchor=north west] {$v_6$};
\draw (-5.9,4.38) node[anchor=north west] {$v_7$};
\draw (-5.9,6.44) node[anchor=north west] {$v_8$};
\draw (7.8,7.8) node[anchor=north west] {$v_1$};
\draw (10.2,6.5) node[anchor=north west] {$v_2$};
\draw (10.0,5.5) node[anchor=north west] {$v_3$};
\draw (8.2,3.6) node[anchor=north west] {$v_4$};
\draw (6.2,2.1) node[anchor=north west] {$v_5$};
\draw (4.1,3.3) node[anchor=north west] {$v_6$};
\draw (3.7,5.64) node[anchor=north west] {$v_7$};
\draw (3.7,7.9) node[anchor=north west] {$v_8$};
\begin{scriptsize}
\draw [fill=black] (-3.4,2.56) circle (2.5pt);
\draw [fill=black] (-1.4,2.56) circle (2.5pt);
\draw [fill=black] (0.014213562373094568,3.9742135623730945) circle (2.5pt);
\draw [fill=black] (0.014213562373095012,5.974213562373094) circle (2.5pt);
\draw [fill=black] (-1.4,7.388427124746189) circle (2.5pt);
\draw [fill=black] (-3.4,7.3884271247461895) circle (2.5pt);
\draw [fill=black] (-4.814213562373094,5.9742135623730945) circle (2.5pt);
\draw [fill=black] (-4.814213562373094,3.9742135623730954) circle (2.5pt);
\end{scriptsize}
\end{tikzpicture}
    \caption{A circular-arc 2-graph with a generalized total-circular ordering of the vertices, where $X_1$= $\{ v_1, v_4, v_6, v_8\}$ , $X_2$=$\{v_2, v_3, v_5, v_7\}$.}
    \label{fig:enter-label-1}
\end{figure}
If we calculate the circular arcs of the given graph corresponding to the total-circular ordering of the vertices as shown in Figure 1, we get the following arcs: $A_{v_1}=[1,1]$, $A_{v_2}=[1,2]$, $A_{v_3}=[8,3]$, $A_{v_4}=[2,4]$, $A_{v_5}=[4,5]$, $A_{v_6}=[5,6]$, $A_{v_7}=[6,7]$, and  $A_{v_8}=[7,8]$.\vspace{.3cm}
\par Before introducing another vertex-ordering characterization, we first define the notion of \emph{almost consecutive ones} in the rows and columns of the adjacency matrix of $r$-partite graphs.  

\textbf{Definition 2.} Let $B = (X_1, X_2, \ldots, X_r, E)$ be an $r$-partite graph and let $A$ be the adjacency matrix of $B$. A row (say, the $i$-th row) is said to have \emph{almost consecutive ones} if, between any two ones in the row, whenever a zero appears at the position $(v_i, v_k)$, the vertices $v_i$ and $v_k$ belong to the same partite set of $B$.\\
\par
\textbf{Definition 3.} Let $B = (X_1, X_2, \ldots, X_r, E)$ be an $r$-partite graph and let $A$ be the adjacency matrix of $B$. A column (say, the $j$-th column) is said to have \emph{almost consecutive ones} if, between any two ones in the column, whenever a zero appears at the position $(v_l, v_j)$, the vertices $v_l$ and $v_j$ belong to the same partite set of $B$.

\par We now introduce another vertex ordering for $r$-partite graphs, referred to as $r$-circular ordering. Using this ordering, we characterize the class of circular-arc $r$-graphs.  

Let $B = (X_1, X_2, \dots, X_r, E)$ be a $r$-partite graph of order $n$. Let 
$v_1, v_2, \dots, v_n$ be an ordering of the vertex set  
$ X = \bigcup_{i=1}^r X_i $  
of $B$. Arrange these vertices on an $n$-hour clock so that the $i$-th vertex is placed at the $i$-th hour mark.  

Let $\textbf{M}$ be the adjacency matrix of $B$, where both rows and columns are indexed according to the increasing order of the vertex indices.  

Consider any row of $\textbf{M}$ (say, the $i$-th row). Define $\mathcal{W}_i$ as the set of 1’s in this row that appear almost consecutively, starting from column $s_i$, where $v_{s_i}$ is the first vertex encountered in the anticlockwise direction from $v_i$, and belongs to a different partite set than $v_i$ (note that if $v_i v_{s_i} \notin E$, then $\mathcal{W}_i = \varnothing$). The sequence continues leftward (wrapping around if necessary) until the last almost consecutive 1 is reached in this manner.  

Similarly, consider any column of $\textbf{M}$ (say, the $j$-th column) corresponding to vertex $v_j$. Define $\mathcal{Q}_j$ as the set of 1’s in this column that appear almost consecutively, starting from row $t_j$, where $v_{t_j}$ is the first vertex encountered in the anticlockwise direction from $v_j$, and belongs to a different partite set than $v_j$ (note that if $v_j v_{t_j} \notin E$, then $\mathcal{Q}_j = \varnothing$). The sequence continues upward (wrapping around if necessary) until the last almost consecutive 1 is reached in this manner.  

An ordering of the vertices of $B$ is called an \textit{$r$-circular ordering} if the sets $\mathcal{W}_i$ and $\mathcal{Q}_j$ together contain all the 1’s of the adjacency matrix $\textbf{M}$.  
\begin{theo}\label{t4}
	An $r$-partite graph $B = (X_1, X_2, \dots, X_r, E)$ is a circular-arc $r$-graph if and only if its vertex set $ X = \bigcup_{i=1}^r X_i $ admits an $r$-circular ordering.
\end{theo}
\begin{proof}
	Necessity: Let $B = (X_1, X_2, \dots, X_r, E)$ be a circular-arc $r$-graph of order $n$. Then there exists a circular-arc model  $\mathcal{A}$=$\{A_v : v \in X = \bigcup_{i=1}^r X_i \}$  
	such that $uv \in E$ if and only if $A_u \cap A_v \neq \varnothing$, where $u$ and $v$ belong to different partite sets.  
	
	Without loss of generality, we may assume that the circular-arc model is chosen so that:  
	\begin{enumerate}
		\item none of its arcs coincides with the entire circle;
		\item all arcs are closed (i.e., each arc contains its endpoints);
		\item no two arcs share the same clockwise endpoint.
	\end{enumerate}
	
	Label the vertices of $B$ as $v_1, v_2, \ldots, v_n$ in the order of increasing clockwise endpoints of their corresponding arcs, and arrange the rows and columns of the adjacency matrix of $B$ according to this vertex ordering. We claim that, under this arrangement, the sets $\mathcal{W}_i$ and $\mathcal{Q}_j$ together contain all the 1’s of the adjacency matrix $\textbf{M}$.
	
	\par Let the \((i, j)\)-th position of the adjacency matrix $\textbf{M}$ contain a \(1\), which implies that $v_i$ is adjacent to $v_j$. Suppose $v_i\in X_{\alpha}$ and $v_j\in X_{\beta}$. If \(i > j\), then based on the ordering of the vertices of \(B\), one of the following must hold:
	\begin{itemize}
		\item  \(v_i v_k \in E\) for all $v_k\notin X_{\alpha}$, where \(k \in \{j + 1, j + 2, \dots, i - 1\}\). In this case, the \(1\) at position \((i, j)\) in the adjacency matrix of \(B\) must be contained in \(\mathcal{W}_i\).
		
		\item \(v_l v_j \in E\) for all $v_l\notin X_{\beta}$, where \(l \in \{i + 1, i + 2, \dots, n, 1, \dots, j - 1\}\). In this case, the \(1\) at position \((i,j)\) in the adjacency matrix of \(B\) must be contained in \(\mathcal{Q}_j\).
	\end{itemize}
	Similarly, if \(i < j\), a parallel argument shows that the \(1\) at position \((i, j)\) must be contained in either \(\mathcal{W}_i\) or \(\mathcal{Q}_j\).
	
	Thus, in every case, either \(\mathcal{W}_i\) or \(\mathcal{Q}_j\) must contain the \(1\) at position \((i, j)\) in the adjacency matrix of the $r$-partite graph \(B\). Therefore, the sets \(\mathcal{W}_i\) and $\mathcal{Q}_j$ \((1 \leq i,j \leq n)\) collectively contain all the \(1\)'s in the adjacency matrix $\textbf{M}$.

	\par Sufficiency: Consider an $r$-partite graph $B = (X_1, X_2, \dots, X_r, E)$, where the vertex set $X = \bigcup_{i=1}^r X_i$  
	is ordered as $v_1, v_2, \dots, v_n$. Place these vertices on an $n$-hour clock so that the $i$-th vertex is positioned at the $i$-th hour marker, and assume this ordering ensures that the sets $\mathcal{W}_i$ and $\mathcal{Q}_j$ together contain all the $1$’s of the adjacency matrix of $B$.  
	
	Let $\mathcal{W}_i$ ($1\leq i\leq n$) start from the position $(i, s_i)$ in the adjacency matrix of $B$ and continue leftward (wrapping around if necessary) until the last almost consecutive $1$ is encountered, terminating at position $(i, p_i)$. Draw an arc $A_i$ on the $n$-hour clock in the clockwise direction starting from $p_i$ and ending at $i$, and associate this arc with vertex $v_i$.  
	
	Suppose $v_i v_j \in E$. Then the position $(i, j)$ in the adjacency matrix contains a $1$. Consequently, either $W_i$ or $\mathcal{Q}_j$ contains this $1$.  
	\begin{itemize}
		\item If $\mathcal{W}_i$ contains this $1$, then the arc $A_i$ contains $j$.  
		\item If $\mathcal{Q}_j$ contains this $1$, then the arc $A_j$ contains $i$.  
	\end{itemize}  
	In either case, $A_i \cap A_j \neq \varnothing$.  
	
	Conversely, suppose $A_i \cap A_j \neq \varnothing$. Then either $A_i$ contains the clockwise endpoint of $A_j$, or $A_j$ contains the clockwise endpoint of $A_i$. This means that either $A_i$ contains $j$ or $A_j$ contains $i$.  
	\begin{itemize}
		\item If $A_i$ contains $j$, then $\mathcal{W}_i$ will include the position $(i, j)$ of the adjacency matrix, which must therefore be $1$, implying $v_i v_j \in E$.  
		\item Similarly, if $A_j$ contains $i$, then $\mathcal{Q}_j$ will include the position $(i, j)$ of the adjacency matrix, which must therefore be $1$, again implying $v_i v_j \in E$.  
	\end{itemize}  
	
	Hence, $v_i$ is adjacent to $v_j$ if and only if $A_i \cap A_j \neq \varnothing$, where $v_i$ and $v_j$ belong to different partite sets.  
	
	Therefore, $B = (X_1, X_2, \dots, X_r, E)$ is a circular-arc $r$-graph.
\end{proof}

\begin{figure}[H]
    \centering
   \begin{tikzpicture}[scale=1., every node/.style={circle, inner sep=1.5pt}]
    \node (v1) at (0,2) [fill=black] {};
    \node[draw=none, fill=none] at (0, 2.2) {$v_1$};
    \node (v4) at (1,2) [fill=black] {};
    \node[draw=none, fill=none] at (1, 2.2) {$v_4$};
    \node (v5) at (2,2) [fill=black] {};
    \node[draw=none, fill=none] at (2, 2.2) {$v_5$};
    \node (v7) at (3,2) [fill=black] {};
     \node[draw=none, fill=none] at (3, 2.2) {$v_7$};
    \node (v8) at (4,2) [fill=black] {};
    \node[draw=none, fill=none] at (4, 2.2) {$v_8$};
    \node (v10) at (5,2) [fill=black] {};
 \node[draw=none, fill=none] at (5, 2.2) {$v_{10}$};

    \node (v2) at (1,0) [fill=red] {};
     \node[draw=none, fill=none] at (1, -.2) {$v_2$};
    \node (v3) at (2,0) [fill=red] {};
     \node[draw=none, fill=none] at (2, -.2) {$v_3$};
    \node (v6) at (3,0) [fill=red] {};
    \node[draw=none, fill=none] at (3, -.2) {$v_6$};
    \node (v9) at (4,0) [fill=red] {};
    \node[draw=none, fill=none] at (4, -.2) {$v_9$};

\draw [line width=.5pt] (v1)--(v2);
\draw [line width=.5pt] (v1)--(v3);
\draw [line width=.5pt] (v1)--(v9);

\draw [line width=.5pt] (v2)--(v4);
\draw [line width=.5pt] (v3)--(v4);
\draw [line width=.5pt] (v3)--(v5);
\draw [line width=.5pt] (v3)--(v10);

\draw [line width=.5pt] (v4)--(v6);
\draw [line width=.5pt] (v5)--(v6);
\draw [line width=.5pt] (v6)--(v8);
\draw [line width=.5pt] (v6)--(v10);
\draw [line width=.5pt] (v7)--(v9);
\draw [line width=.5pt] (v8)--(v9);
\draw [line width=.5pt] (v9)--(v10);

\draw [rotate around={0.:(2.5,2.1)},line width=.5pt] (2.5,2.1) ellipse (3.5cm and 0.35cm);

\draw [rotate around={0.:(2.5,-.1)},line width=.5pt] (2.5,-.1) ellipse (2.5cm and 0.35cm);

\draw (5,0.1) node[anchor=north west,scale=1.5] {$\leftarrow$};
\draw (5.5,0.1) node[anchor=north west,scale=1.] {$X_1$};
\draw (6,2.3) node[anchor=north west,scale=1.5] {$\leftarrow$};
\draw (6.5,2.3) node[anchor=north west,scale=1.] {$X_2$};
\end{tikzpicture}
    \caption{A bipartite graph having an ordering of its vertices: $v_1$, $v_2$, $v_3$, $v_4$, $v_5$, $v_6$, $v_7$, $v_8$, $v_9$, $v_{10}$. It is a 2-circular ordering as shown in the next figure.}
    \label{fig:placeholder}
\end{figure}
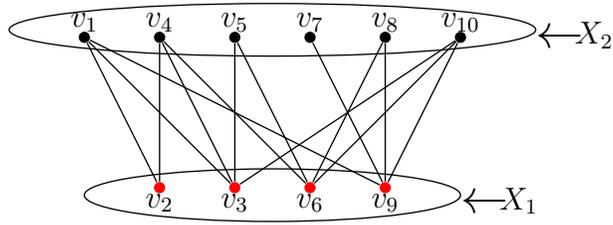

\begin{figure}[H]
    \centering

    \begin{tikzpicture}[line cap=round,line join=round,x=1cm,y=1cm,scale=1.5]
\clip(.5,.9) rectangle (11.2,11.5);
\draw [line width=1.pt] (1,1)-- (11.,1);
\draw [line width=1.pt] (1,1)-- (1.,11);
\draw [line width=1.pt] (1,11)-- (11.,11);
\draw [line width=1.pt] (11,1)-- (11.,11);
\draw [line width=1.pt] (2,1)-- (2.,11);
\draw [line width=1.pt] (3,1)-- (3.,11);
\draw [line width=1.pt] (4,1)-- (4.,11);
\draw [line width=1.pt] (5,1)-- (5.,11);
\draw [line width=1.pt] (6,1)-- (6.,11);
\draw [line width=1.pt] (7,1)-- (7.,11);
\draw [line width=1.pt] (8,1)-- (8.,11);
\draw [line width=1.pt] (9,1)-- (9.,11);
\draw [line width=1.pt] (10,1)-- (10.,11);
\draw [line width=1.pt] (1,2)-- (11.,2);
\draw [line width=1.pt] (1,3)-- (11.,3);
\draw [line width=1.pt] (1,4)-- (11.,4);
\draw [line width=1.pt] (1,5)-- (11.,5);
\draw [line width=1.pt] (1,6)-- (11.,6);
\draw [line width=1.pt] (1,7)-- (11.,7);
\draw [line width=1.pt] (1,8)-- (11.,8);
\draw [line width=1.pt] (1,9)-- (11.,9);
\draw [line width=1.pt] (1,10)-- (11.,10);

\draw (1.3,1.7) node[anchor=north west] {$0$};
\draw (2.3,1.7) node[anchor=north west] {$0$};
\draw (3.3,1.7) node[anchor=north west] {$1$};
\draw (4.3,1.7) node[anchor=north west] {$0$};
\draw (5.3,1.7) node[anchor=north west] {$0$};
\draw (6.3,1.7) node[anchor=north west] {$1$};
\draw (7.3,1.7) node[anchor=north west] {$0$};
\draw (8.3,1.7) node[anchor=north west] {$0$};
\draw (9.3,1.7) node[anchor=north west] {$1$};
\draw (10.3,1.7) node[anchor=north west] {$0$};

\draw (1.3,2.7) node[anchor=north west] {$1$};
\draw (2.3,2.7) node[anchor=north west] {$0$};
\draw (3.3,2.7) node[anchor=north west] {$0$};
\draw (4.3,2.7) node[anchor=north west] {$0$};
\draw (5.3,2.7) node[anchor=north west] {$0$};
\draw (6.3,2.7) node[anchor=north west] {$0$};
\draw (7.3,2.7) node[anchor=north west] {$1$};
\draw (8.3,2.7) node[anchor=north west] {$1$};
\draw (9.3,2.7) node[anchor=north west] {$0$};
\draw (10.3,2.7) node[anchor=north west] {$1$};

\draw (1.3,3.7) node[anchor=north west] {$0$};
\draw (2.3,3.7) node[anchor=north west] {$0$};
\draw (3.3,3.7) node[anchor=north west] {$0$};
\draw (4.3,3.7) node[anchor=north west] {$0$};
\draw (5.3,3.7) node[anchor=north west] {$0$};
\draw (6.3,3.7) node[anchor=north west] {$1$};
\draw (7.3,3.7) node[anchor=north west] {$0$};
\draw (8.3,3.7) node[anchor=north west] {$0$};
\draw (9.3,3.7) node[anchor=north west] {$1$};
\draw (10.3,3.7) node[anchor=north west] {$0$};

\draw (1.3,4.7) node[anchor=north west] {$0$};
\draw (2.3,4.7) node[anchor=north west] {$0$};
\draw (3.3,4.7) node[anchor=north west] {$0$};
\draw (4.3,4.7) node[anchor=north west] {$0$};
\draw (5.3,4.7) node[anchor=north west] {$0$};
\draw (6.3,4.7) node[anchor=north west] {$1$};
\draw (7.3,4.7) node[anchor=north west] {$0$};
\draw (8.3,4.7) node[anchor=north west] {$0$};
\draw (9.3,4.7) node[anchor=north west] {$1$};
\draw (10.3,4.7) node[anchor=north west] {$0$};

\draw (1.3,5.7) node[anchor=north west] {$0$};
\draw (2.3,5.7) node[anchor=north west] {$0$};
\draw (3.3,5.7) node[anchor=north west] {$0$};
\draw (4.3,5.7) node[anchor=north west] {$1$};
\draw (5.3,5.7) node[anchor=north west] {$1$};
\draw (6.3,5.7) node[anchor=north west] {$0$};
\draw (7.3,5.7) node[anchor=north west] {$1$};
\draw (8.3,5.7) node[anchor=north west] {$1$};
\draw (9.3,5.7) node[anchor=north west] {$0$};
\draw (10.3,5.7) node[anchor=north west] {$1$};

\draw (1.3,6.7) node[anchor=north west] {$0$};
\draw (2.3,6.7) node[anchor=north west] {$0$};
\draw (3.3,6.7) node[anchor=north west] {$1$};
\draw (4.3,6.7) node[anchor=north west] {$0$};
\draw (5.3,6.7) node[anchor=north west] {$0$};
\draw (6.3,6.7) node[anchor=north west] {$1$};
\draw (7.3,6.7) node[anchor=north west] {$0$};
\draw (8.3,6.7) node[anchor=north west] {$0$};
\draw (9.3,6.7) node[anchor=north west] {$0$};
\draw (10.3,6.7) node[anchor=north west] {$0$};

\draw (1.3,7.7) node[anchor=north west] {$0$};
\draw (2.3,7.7) node[anchor=north west] {$1$};
\draw (3.3,7.7) node[anchor=north west] {$1$};
\draw (4.3,7.7) node[anchor=north west] {$0$};
\draw (5.3,7.7) node[anchor=north west] {$0$};
\draw (6.3,7.7) node[anchor=north west] {$1$};
\draw (7.3,7.7) node[anchor=north west] {$0$};
\draw (8.3,7.7) node[anchor=north west] {$0$};
\draw (9.3,7.7) node[anchor=north west] {$0$};
\draw (10.3,7.7) node[anchor=north west] {$0$};

\draw (1.3,8.7) node[anchor=north west] {$1$};
\draw (2.3,8.7) node[anchor=north west] {$0$};
\draw (3.3,8.7) node[anchor=north west] {$0$};
\draw (4.3,8.7) node[anchor=north west] {$1$};
\draw (5.3,8.7) node[anchor=north west] {$1$};
\draw (6.3,8.7) node[anchor=north west] {$0$};
\draw (7.3,8.7) node[anchor=north west] {$0$};
\draw (8.3,8.7) node[anchor=north west] {$0$};
\draw (9.3,8.7) node[anchor=north west] {$0$};
\draw (10.3,8.7) node[anchor=north west] {$1$};

\draw (1.3,9.7) node[anchor=north west] {$1$};
\draw (2.3,9.7) node[anchor=north west] {$0$};
\draw (3.3,9.7) node[anchor=north west] {$0$};
\draw (4.3,9.7) node[anchor=north west] {$1$};
\draw (5.3,9.7) node[anchor=north west] {$0$};
\draw (6.3,9.7) node[anchor=north west] {$0$};
\draw (7.3,9.7) node[anchor=north west] {$0$};
\draw (8.3,9.7) node[anchor=north west] {$0$};
\draw (9.3,9.7) node[anchor=north west] {$0$};
\draw (10.3,9.7) node[anchor=north west] {$0$};

\draw (1.3,10.7) node[anchor=north west] {$0$};
\draw (2.3,10.7) node[anchor=north west] {$1$};
\draw (3.3,10.7) node[anchor=north west] {$1$};
\draw (4.3,10.7) node[anchor=north west] {$0$};
\draw (5.3,10.7) node[anchor=north west] {$0$};
\draw (6.3,10.7) node[anchor=north west] {$0$};
\draw (7.3,10.7) node[anchor=north west] {$0$};
\draw (8.3,10.7) node[anchor=north west] {$0$};
\draw (9.3,10.7) node[anchor=north west] {$1$};
\draw (10.3,10.7) node[anchor=north west] {$0$};

\draw (1.3,11.5) node[anchor=north west] {$v_1$};
\draw (2.3,11.5) node[anchor=north west] {$v_2$};
\draw (3.3,11.5) node[anchor=north west] {$v_3$};
\draw (4.3,11.5) node[anchor=north west] {$v_4$};
\draw (5.3,11.5) node[anchor=north west] {$v_5$};
\draw (6.3,11.5) node[anchor=north west] {$v_6$};
\draw (7.3,11.5) node[anchor=north west] {$v_7$};
\draw (8.3,11.5) node[anchor=north west] {$v_8$};
\draw (9.3,11.5) node[anchor=north west] {$v_9$};
\draw (10.3,11.5) node[anchor=north west] {$v_{10}$};

\draw (0.5,10.7) node[anchor=north west] {$v_1$};
\draw (0.5,9.7) node[anchor=north west] {$v_2$};
\draw (0.5,8.7) node[anchor=north west] {$v_3$};
\draw (0.5,7.7) node[anchor=north west] {$v_4$};
\draw (0.5,6.7) node[anchor=north west] {$v_5$};
\draw (0.5,5.7) node[anchor=north west] {$v_6$};
\draw (0.5,4.7) node[anchor=north west] {$v_7$};
\draw (0.5,3.7) node[anchor=north west] {$v_8$};
\draw (0.5,2.7) node[anchor=north west] {$v_9$};
\draw (0.5,1.7) node[anchor=north west] {$v_{10}$};

\draw [rotate around={-90.:(2.5,10.5)},line width=1.pt] (2.5,10.5) ellipse (0.5cm and 0.15cm);
\draw [rotate around={-90.:(3.5,11.5)},line width=1.pt] (3.5,11.5) ellipse (1.5cm and 0.15cm);
\draw [rotate around={0.:(9.5,10.5)},line width=1.pt] (9.5,10.5) ellipse (0.5cm and 0.15cm);
\draw [rotate around={0.:(1.5,9.5)},line width=1.pt] (1.5,9.5) ellipse (0.5cm and 0.15cm);
\draw [rotate around={90.:(4.5,9)},line width=1.pt] (4.5,9) ellipse (1cm and 0.2cm);
\draw [rotate around={90.:(5.5,8.5)},line width=1.pt] (5.5,8.5) ellipse (0.5cm and 0.15cm);
\draw [rotate around={0.:(.5,8.5)},line width=1.pt] (.5,8.5) ellipse (1.5cm and 0.2cm);
\draw [rotate around={180.:(11.5,8.5)},line width=1.pt] (11.5,8.5) ellipse (1.5cm and 0.2cm);
\draw [rotate around={0.:(3,7.5)},line width=1.pt] (3,7.5) ellipse (1.cm and 0.2cm);
\draw [rotate around={90.:(6.5,7.)},line width=1.pt] (6.5,7.) ellipse (1.cm and 0.2cm);
\draw [rotate around={0.:(3.5,6.5)},line width=1.pt] (3.5,6.5) ellipse (.5cm and 0.15cm);
\draw [rotate around={0.:(5,5.5)},line width=1.pt] (5,5.5) ellipse (1.cm and 0.2cm);
\draw [rotate around={90.:(7.5,5.5)},line width=1.pt] (7.5,5.5) ellipse (.5cm and 0.15cm);
\draw [rotate around={90.:(8.5,5.5)},line width=1.pt] (8.5,5.5) ellipse (.5cm and 0.15cm);
\draw [rotate around={90.:(10.5,5.5)},line width=1.pt] (10.5,5.5) ellipse (3.5cm and 0.25cm);
\draw [rotate around={0.:(6.5,1.5)},line width=1.pt] (6.5,1.5) ellipse (3.5cm and 0.25cm);
\draw [rotate around={0.:(8,2.5)},line width=1.pt] (8,2.5) ellipse (1cm and 0.2cm);
\draw [rotate around={90.:(9.5,4)},line width=1.pt] (9.5,4) ellipse (1cm and 0.2cm);
\draw [rotate around={0.:(6.5,4.5)},line width=1.pt] (6.5,4.5) ellipse (.5cm and 0.15cm);
\draw [rotate around={0.:(6.5,3.5)},line width=1.pt] (6.5,3.5) ellipse (.5cm and 0.2cm);
\draw [rotate around={90.:(3.5,.5)},line width=1.pt] (3.5,.5) ellipse (1.5cm and 0.25cm);
\draw [rotate around={90.:(1.5,2.5)},line width=1.pt] (1.5,2.5) ellipse (.5cm and 0.15cm);

\draw (9.2,11) node[anchor=north west,scale=.8] {$W_1$};
\draw (9.2,10.4) node[anchor=north west,scale=1.5] {$\leftarrow$};
\draw (2.,10.5) node[anchor=north west,scale=.8] {$Q_2$};
\draw (2.5,10.7) node[anchor=north west,scale=1.5] {$\uparrow$};
\draw (3.,10.5) node[anchor=north west,scale=.8] {$Q_3$};
\draw (3.5,10.7) node[anchor=north west,scale=1.5] {$\uparrow$};
\draw (1.2,10) node[anchor=north west,scale=.8] {$W_2$};
\draw (1.3,9.4) node[anchor=north west,scale=1.5] {$\leftarrow$};
\draw (4.,8.5) node[anchor=north west,scale=.8] {$Q_4$};
\draw (4.5,8.7) node[anchor=north west,scale=1.5] {$\uparrow$};
\draw (5.,8.5) node[anchor=north west,scale=.8] {$Q_4$};
\draw (5.5,8.7) node[anchor=north west,scale=1.5] {$\uparrow$};
\draw (1.2,9) node[anchor=north west,scale=.8] {$W_3$};
\draw (1.3,8.4) node[anchor=north west,scale=1.5] {$\leftarrow$};
\draw (3.2,8) node[anchor=north west,scale=.8] {$W_4$};
\draw (3.3,7.4) node[anchor=north west,scale=1.5] {$\leftarrow$};
\draw (3.2,7) node[anchor=north west,scale=.8] {$W_5$};
\draw (3.3,6.4) node[anchor=north west,scale=1.5] {$\leftarrow$};
\draw (6.,6.4) node[anchor=north west,scale=.8] {$Q_6$};
\draw (6.5,6.6) node[anchor=north west,scale=1.5] {$\uparrow$};
\draw (5.2,6) node[anchor=north west,scale=.8] {$W_6$};
\draw (5.3,5.4) node[anchor=north west,scale=1.5] {$\leftarrow$};
\draw (7.,5.4) node[anchor=north west,scale=.8] {$Q_7$};
\draw (7.5,5.6) node[anchor=north west,scale=1.5] {$\uparrow$};
\draw (8.,5.4) node[anchor=north west,scale=.8] {$Q_8$};
\draw (8.5,5.6) node[anchor=north west,scale=1.5] {$\uparrow$};
\draw (6.2,5) node[anchor=north west,scale=.8] {$W_7$};
\draw (6.3,4.4) node[anchor=north west,scale=1.5] {$\leftarrow$};
\draw (6.2,4) node[anchor=north west,scale=.8] {$W_8$};
\draw (6.3,3.4) node[anchor=north west,scale=1.5] {$\leftarrow$};
\draw (8.2,3) node[anchor=north west,scale=.8] {$W_9$};
\draw (8.3,2.4) node[anchor=north west,scale=1.5] {$\leftarrow$};
\draw (9.2,2) node[anchor=north west,scale=.8] {$W_{10}$};
\draw (9.3,1.4) node[anchor=north west,scale=1.5] {$\leftarrow$};
\draw (9.,3.4) node[anchor=north west,scale=.8] {$Q_9$};
\draw (9.5,3.6) node[anchor=north west,scale=1.5] {$\uparrow$};
\draw (10.,2.4) node[anchor=north west,scale=.8] {$Q_{10}$};
\draw (10.5,2.6) node[anchor=north west,scale=1.5] {$\uparrow$};
\draw (1.,2.4) node[anchor=north west,scale=.8] {$Q_{1}$};
\draw (1.5,2.6) node[anchor=north west,scale=1.5] {$\uparrow$};
\end{tikzpicture}
    \caption{The adjacency matrix of the bigraph in Figure 2, where the rows and columns are arranged according to the increasing order of their indices, and corresponding $W_i$'s and $Q_j$'s.}
    \label{}
\end{figure}

\noindent{}
Note that the zeros inside the elliptic regions of $W_{10}$ and $Q_{10}$, as shown in 
Figure~3, do not belong to the sets $W_{10}$ and $Q_{10}$. They simply indicate that 
the 1’s in $W_{10}$ and $Q_{10}$ are not strictly consecutive, but rather almost consecutive.\\

The circular-arc representation of the bigraph of Figure 2 is the following: $A_{v_1}=[9,1]$, $A_{v_2}=[1,2]$, $A_{v_3}=[10,3]$, $A_{v_4}=[2,4]$, $A_{v_5}=[3,5]$, $A_{v_6}=[4,6]$, $A_{v_7}=[6,7]$, $A_{v_8}=[6,8]$, $A_{v_9}=[7,9]$ and, $A_{v_{10}}=[3,10]$.\\

 \par Pavol Hell and Jing Huang \cite{hell} characterized interval bigraphs using forbidden patterns with respect to a specific vertex ordering in the following theorem.
\begin{theo}[\cite{hell}]
   \textit{ Let $H$ be a bipartite graph with bipartition $(X,Y)$. Then the following statements are equivalent}:
    \begin{itemize}
        \item\textit{ $H$ is an interval bigraph};
        \item\textit{ the vertices of $H$ can be ordered $v_1$, $v_2$, ..., $v_n$, so that there do not exist $a<b<c$ in the configurations in Figure 4. (Black vertices are in $X$, red vertices in $Y$, or conversely, and all edges not shown are absent.)}
        \item \textit{the vertices of $H$ can be ordered $v_1$, $v_2$, ..., $v_n$, so that there do not exist $a<b<c<d$ in the configurations in Figure 5}.
    \end{itemize}
\end{theo}
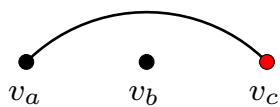
\begin{figure}[H]
    \centering
    \begin{tikzpicture}[line cap=round,line join=round,x=1.0cm,y=1.0cm,scale=.8]
\clip(1.,-1) rectangle (7.,3.);
\draw [shift={(4.,-1.)},line width=1.pt]  plot[domain=0.7853981633974483:2.356194490192345,variable=\t]({1.*2.8284271247461903*cos(\t r)+0.*2.8284271247461903*sin(\t r)},{0.*2.8284271247461903*cos(\t r)+1.*2.8284271247461903*sin(\t r)});
\begin{scriptsize}
\draw [fill=black] (2.,1.) circle (3.5pt);
\draw [fill=black] (4.,1.) circle (3.5pt);
\draw [fill=red] (6.,1.) circle (3.5pt);
\draw (1.5,.8) node[anchor=north west,scale=1.5] {$v_a$};
\draw (3.5,.8) node[anchor=north west,scale=1.5] {$v_b$};
\draw (5.5,.8) node[anchor=north west,scale=1.5] {$v_c$};
\end{scriptsize}
\end{tikzpicture}
    \caption{Forbidden pattern.}
    \label{fig:enter-label-3}
\end{figure}
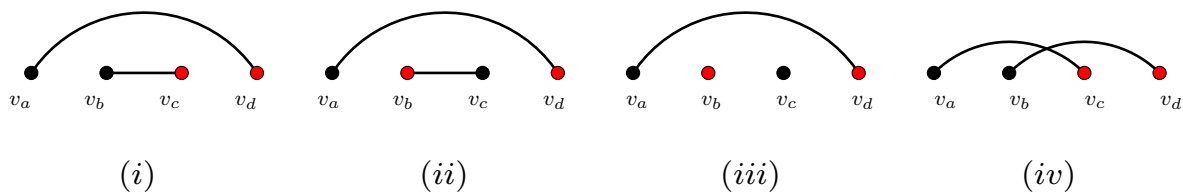
\begin{figure}[H]
    \centering
    \begin{tikzpicture}[line cap=round,line join=round,x=1.0cm,y=1.0cm, scale=1]
\clip(1.,0.) rectangle (19.,4.);
\draw [line width=1.pt] (3.,2.)-- (4.,2.);
\draw [shift={(3.5,1.)},line width=1.pt]  plot[domain=0.5880026035475675:2.5535900500422257,variable=\t]({1.*1.8027756377319948*cos(\t r)+0.*1.8027756377319948*sin(\t r)},{0.*1.8027756377319948*cos(\t r)+1.*1.8027756377319948*sin(\t r)});
\draw [line width=1.pt] (7.,2.)-- (8.,2.);
\draw [shift={(7.5,1.)},line width=1.pt]  plot[domain=0.5880026035475675:2.5535900500422257,variable=\t]({1.*1.8027756377319948*cos(\t r)+0.*1.8027756377319948*sin(\t r)},{0.*1.8027756377319948*cos(\t r)+1.*1.8027756377319948*sin(\t r)});
\draw [shift={(11.5,1.)},line width=1.pt]  plot[domain=0.5880026035475675:2.5535900500422257,variable=\t]({1.*1.8027756377319948*cos(\t r)+0.*1.8027756377319948*sin(\t r)},{0.*1.8027756377319948*cos(\t r)+1.*1.8027756377319948*sin(\t r)});
\draw [shift={(15.,1.)},line width=1.pt]  plot[domain=0.7853981633974483:2.356194490192345,variable=\t]({1.*1.4142135623730951*cos(\t r)+0.*1.4142135623730951*sin(\t r)},{0.*1.4142135623730951*cos(\t r)+1.*1.4142135623730951*sin(\t r)});
\draw [shift={(16.,1.)},line width=1.pt]  plot[domain=0.7853981633974483:2.356194490192345,variable=\t]({1.*1.4142135623730951*cos(\t r)+0.*1.4142135623730951*sin(\t r)},{0.*1.4142135623730951*cos(\t r)+1.*1.4142135623730951*sin(\t r)});
\begin{scriptsize}
\draw [fill=black] (2.,2.) circle (2.5pt);
\draw [fill=black] (3.,2.) circle (2.5pt);
\draw [fill=red] (4.,2.) circle (2.5pt);
\draw [fill=red] (5.,2.) circle (2.5pt);
\draw [fill=black] (6.,2.) circle (2.5pt);
\draw [fill=red] (7.,2.) circle (2.5pt);
\draw [fill=black] (8.,2.) circle (2.5pt);
\draw [fill=red] (9.,2.) circle (2.5pt);
\draw [fill=black] (10.,2.) circle (2.5pt);
\draw [fill=red] (11.,2.) circle (2.5pt);
\draw [fill=black] (12.,2.) circle (2.5pt);
\draw [fill=red] (13.,2.) circle (2.5pt);
\draw [fill=black] (14.,2.) circle (2.5pt);
\draw [fill=black] (15.,2.) circle (2.5pt);
\draw [fill=red] (16.,2.) circle (2.5pt);
\draw [fill=red] (17.,2.) circle (2.5pt);
\draw (1.6,1.8) node[anchor=north west,scale=1.] {$v_a$};
\draw (2.6,1.8) node[anchor=north west,scale=1.] {$v_b$};
\draw (3.6,1.8) node[anchor=north west,scale=1.] {$v_c$};
\draw (4.6,1.8) node[anchor=north west,scale=1.] {$v_d$};
\draw (3,1) node[anchor=north west,scale=1.5] {$(i)$};

\draw (5.7,1.8) node[anchor=north west,scale=1.] {$v_a$};
\draw (6.7,1.8) node[anchor=north west,scale=1.] {$v_b$};
\draw (7.7,1.8) node[anchor=north west,scale=1.] {$v_c$};
\draw (8.7,1.8) node[anchor=north west,scale=1.] {$v_d$};
\draw (7,1) node[anchor=north west,scale=1.5] {$(ii)$};

\draw (9.8,1.8) node[anchor=north west,scale=1.] {$v_a$};
\draw (10.8,1.8) node[anchor=north west,scale=1.] {$v_b$};
\draw (11.8,1.8) node[anchor=north west,scale=1.] {$v_c$};
\draw (12.8,1.8) node[anchor=north west,scale=1.] {$v_d$};
\draw (11,1) node[anchor=north west,scale=1.5] {$(iii)$};

\draw (13.9,1.8) node[anchor=north west,scale=1.] {$v_a$};
\draw (14.9,1.8) node[anchor=north west,scale=1.] {$v_b$};
\draw (15.9,1.8) node[anchor=north west,scale=1.] {$v_c$};
\draw (16.9,1.8) node[anchor=north west,scale=1.] {$v_d$};
\draw (15,1) node[anchor=north west,scale=1.5] {$(iv)$};

\end{scriptsize}
\end{tikzpicture}
    \caption{Forbidden patterns.}
    \label{fig:enter-label-4}
\end{figure}
Motivated by the above result of Hell and Huang \cite{hell}, Paul and Das \cite{paul_das}  provide a characterization of circular-arc bigraphs in terms of forbidden patterns.
\begin{theo}[\cite{paul_das}]
   \textit{ Let $G$ be a bipartite graph with bipartition $(X,Y)$. Then the following statements are equivalent}:
    \begin{itemize}
        \item \textit{$G$ is a circular-arc bigraph;} 
        \item \textit{The vertices of $G$ can be ordered $v_1$, $v_2$, $v_3$,..., $v_n$, so that there do not exist $i<j<k<l$  in the configurations in Figure 6}. (Black vertices are in $X$, red vertices in $Y$, or conversely, and all edges not shown are absent.)
    \end{itemize}
\end{theo}
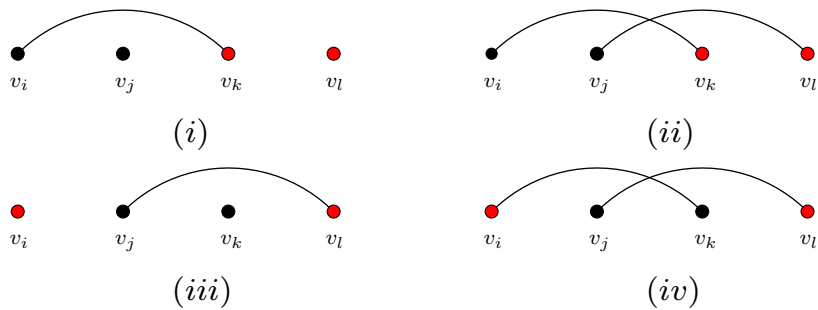
\begin{figure}[H]
    \centering
    \begin{tikzpicture}[line cap=round,line join=round,x=1.0cm,y=1.0cm,scale=.7]
\clip(1.,-1) rectangle (18.,6.);
\draw [shift={(4.,2.)},line width=.5pt]  plot[domain=0.7853981633974483:2.356194490192345,variable=\t]({1.*2.8284271247461903*cos(\t r)+0.*2.8284271247461903*sin(\t r)},{0.*2.8284271247461903*cos(\t r)+1.*2.8284271247461903*sin(\t r)});
\draw [shift={(13.,2.)},line width=.5pt]  plot[domain=0.7853981633974483:2.356194490192345,variable=\t]({1.*2.8284271247461903*cos(\t r)+0.*2.8284271247461903*sin(\t r)},{0.*2.8284271247461903*cos(\t r)+1.*2.8284271247461903*sin(\t r)});
\draw [shift={(15.,2.)},line width=.5pt]  plot[domain=0.7853981633974483:2.356194490192345,variable=\t]({1.*2.8284271247461903*cos(\t r)+0.*2.8284271247461903*sin(\t r)},{0.*2.8284271247461903*cos(\t r)+1.*2.8284271247461903*sin(\t r)});
\draw [shift={(6.,-1.)},line width=.5pt]  plot[domain=0.7853981633974483:2.356194490192345,variable=\t]({1.*2.8284271247461903*cos(\t r)+0.*2.8284271247461903*sin(\t r)},{0.*2.8284271247461903*cos(\t r)+1.*2.8284271247461903*sin(\t r)});
\draw [shift={(13.,-1.)},line width=.5pt]  plot[domain=0.7853981633974483:2.356194490192345,variable=\t]({1.*2.8284271247461903*cos(\t r)+0.*2.8284271247461903*sin(\t r)},{0.*2.8284271247461903*cos(\t r)+1.*2.8284271247461903*sin(\t r)});
\draw [shift={(15.,-1.)},line width=.5pt]  plot[domain=0.7853981633974483:2.356194490192345,variable=\t]({1.*2.8284271247461903*cos(\t r)+0.*2.8284271247461903*sin(\t r)},{0.*2.8284271247461903*cos(\t r)+1.*2.8284271247461903*sin(\t r)});
\begin{scriptsize}
\draw [fill=black] (2.,4.) circle (3.5pt);
\draw (1.7,3.7) node[anchor=north west,scale=1] {$v_i$};
\draw [fill=black] (4.,4.) circle (3.5pt);
\draw (3.7,3.7) node[anchor=north west,scale=1] {$v_j$};
\draw [fill=red] (6.,4.) circle (3.5pt);
\draw (5.7,3.7) node[anchor=north west,scale=1] {$v_k$};
\draw [fill=red] (8.,4.) circle (3.5pt);
\draw (7.7,3.7) node[anchor=north west,scale=1] {$v_l$};
\draw (4.7,3) node[anchor=north west,scale=1.5] {$(i)$};
\draw [fill=black] (11.,4.) circle (3pt);
\draw (10.7,3.7) node[anchor=north west,scale=1] {$v_i$};
\draw [fill=black] (13.,4.) circle (3.5pt);
\draw (12.7,3.7) node[anchor=north west,scale=1] {$v_j$};
\draw [fill=red] (15.,4.) circle (3.5pt);
\draw (14.7,3.7) node[anchor=north west,scale=1] {$v_k$};
\draw [fill=red] (17.,4.) circle (3.5pt);
\draw (16.7,3.7) node[anchor=north west,scale=1] {$v_l$};
\draw [fill=red] (2.,1.) circle (3.5pt);
\draw (13.7,3) node[anchor=north west,scale=1.5] {$
(ii)$};
\draw (1.7,.7) node[anchor=north west,scale=1] {$v_i$};
\draw [fill=black] (4.,1.) circle (3.5pt);
\draw (3.7,.7) node[anchor=north west,scale=1] {$v_j$};
\draw [fill=black] (6.,1.) circle (3.5pt);
\draw (5.7,.7) node[anchor=north west,scale=1] {$v_k$};
\draw [fill=red] (8.,1.) circle (3.5pt);
\draw (7.7,.7) node[anchor=north west,scale=1] {$v_l$};
\draw [fill=red] (11.,1.) circle (3.5pt);
\draw (4.7,0) node[anchor=north west,scale=1.5] {$
(iii)$};
\draw (10.7,.7) node[anchor=north west,scale=1] {$v_i$};
\draw [fill=black] (13.,1.) circle (3.5pt);
\draw (12.7,.7) node[anchor=north west,scale=1] {$v_j$};
\draw [fill=black] (15.,1.) circle (3.5pt);
\draw (14.7,.7) node[anchor=north west,scale=1] {$v_k$};
\draw [fill=red] (17.,1.) circle (3.5pt);
\draw (16.7,.7) node[anchor=north west,scale=1] {$v_l$};
\draw (13.7,0) node[anchor=north west,scale=1.5] {$
(iv)$};
\end{scriptsize}
\end{tikzpicture}
    \caption{Forbidden patterns.}
    \label{fig:enter-label-5}
\end{figure}
It is a natural and general question whether the classes of circular-arc $r$-graphs ($r \geq 3$) can be characterized by a finite collection of forbidden patterns with respect to some specific ordering of their vertices. In the following theorems, we address and resolve this question.
\begin{theo}
Let $B = (X_1, X_2, X_3, E)$ be a $3$-partite graph. Then the following 
statements are equivalent:
\begin{enumerate}
    \item $B$ is a circular-arc $3$-graph;
    \item The vertices of $B$ can be ordered $v_1, v_2, \ldots, v_n$ such that no indices 
    $i < j < k < \ell$ occur in the configurations shown in Figure~7. 
    (Here, different colors indicate that vertices of different colors belong to 
    different partite sets; moreover, all edges not explicitly drawn are absent.)
\end{enumerate}
\end{theo}

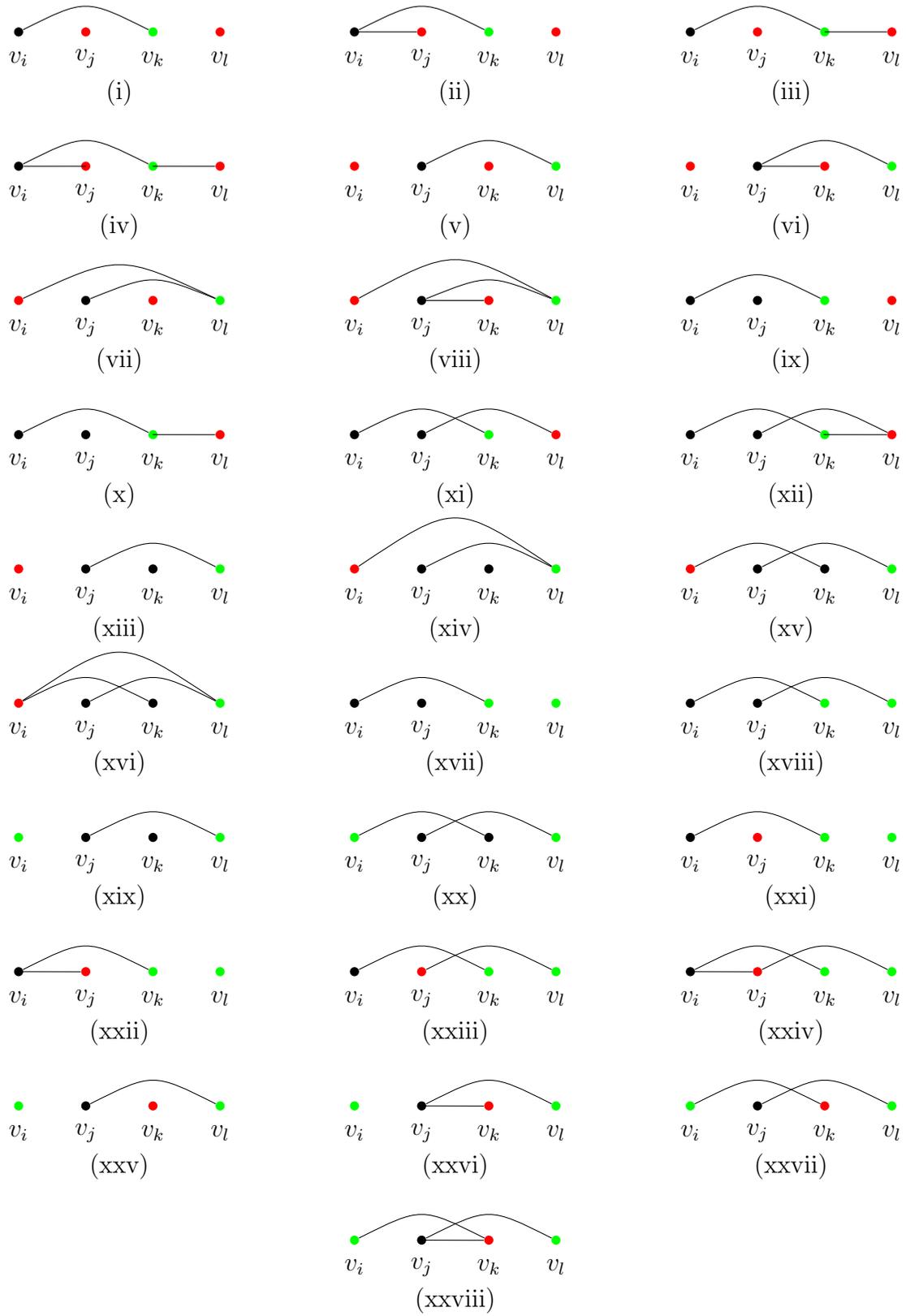
\begin{figure}[H]
\centering

\begin{tikzpicture}[scale=1.1, every node/.style={circle, inner sep=1.5pt}]

\node (vi1) at (0, 0) [fill=black] {};
\node (vj1) at (1, 0) [fill=red] {};
\node (vk1) at (2, 0) [fill=green] {};
\node (vl1) at (3, 0) [fill=red] {};

\draw (vi1) .. controls (1, .5) .. (vk1);

\node[draw=none, fill=none] at (0, -0.4) {$v_i$};
\node[draw=none, fill=none] at (1, -0.4) {$v_j$};
\node[draw=none, fill=none] at (2, -0.4) {$v_k$};
\node[draw=none, fill=none] at (3, -0.4) {$v_l$};

\node[draw=none, fill=none] at (1.5, -0.9) {(i)};
\node (vi2) at (5, 0) [fill=black] {};
\node (vj2) at (6, 0) [fill=red] {};
\node (vk2) at (7, 0) [fill=green] {};
\node (vl2) at (8, 0) [fill=red] {};

\draw (vi2) .. controls (6, .5) .. (vk2);
\draw (vi2) .. controls (5, 0 ).. (vj2);
\node[draw=none, fill=none] at (5, -0.4) {$v_i$};
\node[draw=none, fill=none] at (6, -0.4) {$v_j$};
\node[draw=none, fill=none] at (7, -0.4) {$v_k$};
\node[draw=none, fill=none] at (8, -0.4) {$v_l$};

\node[draw=none, fill=none] at (6.5, -0.9) {(ii)};

\node (vi3) at (10, 0) [fill=black] {};
\node (vj3) at (11, 0) [fill=red] {};
\node (vk3) at (12, 0) [fill=green] {};
\node (vl3) at (13, 0) [fill=red] {};

\draw (vi3) .. controls (11, .5) .. (vk3);
\draw (vk3) .. controls (12,0 ).. (vl3);
\node[draw=none, fill=none] at (10, -0.4) {$v_i$};
\node[draw=none, fill=none] at (11, -0.4) {$v_j$};
\node[draw=none, fill=none] at (12, -0.4) {$v_k$};
\node[draw=none, fill=none] at (13, -0.4) {$v_l$};

\node[draw=none, fill=none] at (11.5, -0.9) {(iii)};

\node (vi4) at (0, -2) [fill=black] {};
\node (vj4) at (1, -2) [fill=red] {};
\node (vk4) at (2, -2) [fill=green] {};
\node (vl4) at (3, -2) [fill=red] {};

\draw (vi4) .. controls (1, -1.5) .. (vk4);
\draw (vi4) .. controls (1, -2) .. (vj4);
\draw (vk4) .. controls (2, -2) .. (vl4);

\node[draw=none, fill=none] at (0, -2.4) {$v_i$};
\node[draw=none, fill=none] at (1, -2.4) {$v_j$};
\node[draw=none, fill=none] at (2, -2.4) {$v_k$};
\node[draw=none, fill=none] at (3, -2.4) {$v_l$};

\node[draw=none, fill=none] at (1.5, -2.9) {(iv)};

\node (vi5) at (5, -2) [fill=red] {};
\node (vj5) at (6, -2) [fill=black] {};
\node (vk5) at (7, -2) [fill=red] {};
\node (vl5) at (8, -2) [fill=green] {};

\draw (vj5) .. controls (7, -1.5) .. (vl5);

\node[draw=none, fill=none] at (5, -2.4) {$v_i$};
\node[draw=none, fill=none] at (6, -2.4) {$v_j$};
\node[draw=none, fill=none] at (7, -2.4) {$v_k$};
\node[draw=none, fill=none] at (8, -2.4) {$v_l$};

\node[draw=none, fill=none] at (6.5, -2.9) {(v)};

\node (vi6) at (10, -2) [fill=red] {};
\node (vj6) at (11, -2) [fill=black] {};
\node (vk6) at (12, -2) [fill=red] {};
\node (vl6) at (13, -2) [fill=green] {};

\draw (vj6) .. controls (11, -2) .. (vk6);
\draw (vj6) .. controls (12, -1.5) .. (vl6);

\node[draw=none, fill=none] at (10, -2.4) {$v_i$};
\node[draw=none, fill=none] at (11, -2.4) {$v_j$};
\node[draw=none, fill=none] at (12, -2.4) {$v_k$};
\node[draw=none, fill=none] at (13, -2.4) {$v_l$};

\node[draw=none, fill=none] at (11.5, -2.9) {(vi)};

\node (vi7) at (0, -4) [fill=red] {};
\node (vj7) at (1, -4) [fill=black] {};
\node (vk7) at (2, -4) [fill=red] {};
\node (vl7) at (3, -4) [fill=green] {};

\draw (vi7) .. controls (1.5, -3.3) .. (vl7);
\draw (vj7) .. controls (2, -3.6) .. (vl7);

\node[draw=none, fill=none] at (0, -4.4) {$v_i$};
\node[draw=none, fill=none] at (1, -4.4) {$v_j$};
\node[draw=none, fill=none] at (2, -4.4) {$v_k$};
\node[draw=none, fill=none] at (3, -4.4) {$v_l$};

\node[draw=none, fill=none] at (1.5, -4.9) {(vii)};

\node (vi8) at (5, -4) [fill=red] {};
\node (vj8) at (6, -4) [fill=black] {};
\node (vk8) at (7, -4) [fill=red] {};
\node (vl8) at (8, -4) [fill=green] {};

\draw (vi8) .. controls (6.5, -3.2) .. (vl8);
\draw (vj8) .. controls (7, -3.6) .. (vl8);
\draw (vj8) .. controls (6, -4) .. (vk8);

\node[draw=none, fill=none] at (5, -4.4) {$v_i$};
\node[draw=none, fill=none] at (6, -4.4) {$v_j$};
\node[draw=none, fill=none] at (7, -4.4) {$v_k$};
\node[draw=none, fill=none] at (8, -4.4) {$v_l$};

\node[draw=none, fill=none] at (6.5, -4.9) {(viii)};

\node (vi9) at (10, -4) [fill=black] {};
\node (vj9) at (11, -4) [fill=black] {};
\node (vk9) at (12, -4) [fill=green] {};
\node (vl9) at (13, -4) [fill=red] {};

\draw (vi9) .. controls (11, -3.5) .. (vk9);

\node[draw=none, fill=none] at (10, -4.4) {$v_i$};
\node[draw=none, fill=none] at (11, -4.4) {$v_j$};
\node[draw=none, fill=none] at (12, -4.4) {$v_k$};
\node[draw=none, fill=none] at (13, -4.4) {$v_l$};

\node[draw=none, fill=none] at (11.5, -4.9) {(ix)};

\node (vi10) at (0, -6) [fill=black] {};
\node (vj10) at (1, -6) [fill=black] {};
\node (vk10) at (2, -6) [fill=green] {};
\node (vl10) at (3, -6) [fill=red] {};

\draw (vi10) .. controls (1, -5.5) .. (vk10);
\draw (vk10) .. controls (2, -6) .. (vl10);

\node[draw=none, fill=none] at (0, -6.4) {$v_i$};
\node[draw=none, fill=none] at (1, -6.4) {$v_j$};
\node[draw=none, fill=none] at (2, -6.4) {$v_k$};
\node[draw=none, fill=none] at (3, -6.4) {$v_l$};

\node[draw=none, fill=none] at (1.5, -6.9) {(x)};

\node (vi11) at (5, -6) [fill=black] {};
\node (vj11) at (6, -6) [fill=black] {};
\node (vk11) at (7, -6) [fill=green] {};
\node (vl11) at (8, -6) [fill=red] {};

\draw (vi11) .. controls (6, -5.5) .. (vk11);
\draw (vj11) .. controls (7, -5.5) .. (vl11);

\node[draw=none, fill=none] at (5, -6.4) {$v_i$};
\node[draw=none, fill=none] at (6, -6.4) {$v_j$};
\node[draw=none, fill=none] at (7, -6.4) {$v_k$};
\node[draw=none, fill=none] at (8, -6.4) {$v_l$};

\node[draw=none, fill=none] at (6.5, -6.9) {(xi)};

\node (vi12) at (10, -6) [fill=black] {};
\node (vj12) at (11, -6) [fill=black] {};
\node (vk12) at (12, -6) [fill=green] {};
\node (vl12) at (13, -6) [fill=red] {};

\draw (vi12) .. controls (11, -5.5) .. (vk12);
\draw (vj12) .. controls (12, -5.5) .. (vl12);
\draw (vk12) .. controls (12, -6) .. (vl12);

\node[draw=none, fill=none] at (10, -6.4) {$v_i$};
\node[draw=none, fill=none] at (11, -6.4) {$v_j$};
\node[draw=none, fill=none] at (12, -6.4) {$v_k$};
\node[draw=none, fill=none] at (13, -6.4) {$v_l$};

\node[draw=none, fill=none] at (11.5, -6.9) {(xii)};

\node (vi13) at (0, -8) [fill=red] {};
\node (vj13) at (1, -8) [fill=black] {};
\node (vk13) at (2, -8) [fill=black] {};
\node (vl13) at (3, -8) [fill=green] {};

\draw (vj13) .. controls (2, -7.5) .. (vl13);

\node[draw=none, fill=none] at (0, -8.4) {$v_i$};
\node[draw=none, fill=none] at (1, -8.4) {$v_j$};
\node[draw=none, fill=none] at (2, -8.4) {$v_k$};
\node[draw=none, fill=none] at (3, -8.4) {$v_l$};

\node[draw=none, fill=none] at (1.5, -8.9) {(xiii)};

\node (vi14) at (5, -8) [fill=red] {};
\node (vj14) at (6, -8) [fill=black] {};
\node (vk14) at (7, -8) [fill=black] {};
\node (vl14) at (8, -8) [fill=green] {};

\draw (vi14) .. controls (6.5, -7.) .. (vl14);
\draw (vj14) .. controls (7, -7.5) .. (vl14);

\node[draw=none, fill=none] at (5, -8.4) {$v_i$};
\node[draw=none, fill=none] at (6, -8.4) {$v_j$};
\node[draw=none, fill=none] at (7, -8.4) {$v_k$};
\node[draw=none, fill=none] at (8, -8.4) {$v_l$};

\node[draw=none, fill=none] at (6.5, -8.9) {(xiv)};

\node (vi15) at (10, -8) [fill=red] {};
\node (vj15) at (11, -8) [fill=black] {};
\node (vk15) at (12, -8) [fill=black] {};
\node (vl15) at (13, -8) [fill=green] {};

\draw (vi15) .. controls (11, -7.5) .. (vk15);
\draw (vj15) .. controls (12, -7.5) .. (vl15);

\node[draw=none, fill=none] at (10, -8.4) {$v_i$};
\node[draw=none, fill=none] at (11, -8.4) {$v_j$};
\node[draw=none, fill=none] at (12, -8.4) {$v_k$};
\node[draw=none, fill=none] at (13, -8.4) {$v_l$};

\node[draw=none, fill=none] at (11.5, -8.9) {(xv)};

\node (vi16) at (0, -10) [fill=red] {};
\node (vj16) at (1, -10) [fill=black] {};
\node (vk16) at (2, -10) [fill=black] {};
\node (vl16) at (3, -10) [fill=green] {};

\draw (vi16) .. controls (1.5, -9.) .. (vl16);
\draw (vi16) .. controls (1, -9.5) .. (vk16);
\draw (vj16) .. controls (2., -9.5) .. (vl16);

\node[draw=none, fill=none] at (0, -10.4) {$v_i$};
\node[draw=none, fill=none] at (1, -10.4) {$v_j$};
\node[draw=none, fill=none] at (2, -10.4) {$v_k$};
\node[draw=none, fill=none] at (3, -10.4) {$v_l$};

\node[draw=none, fill=none] at (1.5, -10.9) {(xvi)};

\node (vi17) at (5, -10) [fill=black] {};
\node (vj17) at (6, -10) [fill=black] {};
\node (vk17) at (7, -10) [fill=green] {};
\node (vl17) at (8,-10) [fill=green] {};

\draw (vi17) .. controls (6, -9.5) .. (vk17);

\node[draw=none, fill=none] at (5, -10.4) {$v_i$};
\node[draw=none, fill=none] at (6, -10.4) {$v_j$};
\node[draw=none, fill=none] at (7, -10.4) {$v_k$};
\node[draw=none, fill=none] at (8, -10.4) {$v_l$};

\node[draw=none, fill=none] at (6.5, -10.9) {(xvii)};

\node (vi18) at (10, -10) [fill=black] {};
\node (vj18) at (11, -10) [fill=black] {};
\node (vk18) at (12, -10) [fill=green] {};
\node (vl18) at (13,-10) [fill=green] {};

\draw (vi18) .. controls (11., -9.5) .. (vk18);
\draw (vj18) .. controls (12., -9.5) .. (vl18);

\node[draw=none, fill=none] at (10, -10.4) {$v_i$};
\node[draw=none, fill=none] at (11, -10.4) {$v_j$};
\node[draw=none, fill=none] at (12, -10.4) {$v_k$};
\node[draw=none, fill=none] at (13, -10.4) {$v_l$};

\node[draw=none, fill=none] at (11.5, -10.9) {(xviii)};

\node (vi19) at (0, -12) [fill=green] {};
\node (vj19) at (1, -12) [fill=black] {};
\node (vk19) at (2, -12) [fill=black] {};
\node (vl19) at (3, -12) [fill=green] {};

\draw (vj19) .. controls (2, -11.5) .. (vl19);

\node[draw=none, fill=none] at (0, -12.4) {$v_i$};
\node[draw=none, fill=none] at (1, -12.4) {$v_j$};
\node[draw=none, fill=none] at (2, -12.4) {$v_k$};
\node[draw=none, fill=none] at (3, -12.4) {$v_l$};

\node[draw=none, fill=none] at (1.5, -12.9) {(xix)};

\node (vi20) at (5, -12) [fill=green] {};
\node (vj20) at (6, -12) [fill=black] {};
\node (vk20) at (7, -12) [fill=black] {};
\node (vl20) at (8, -12) [fill=green] {};

\draw (vi20) .. controls (6, -11.5) .. (vk20);
\draw (vj20) .. controls (7, -11.5) .. (vl20);

\node[draw=none, fill=none] at (5, -12.4) {$v_i$};
\node[draw=none, fill=none] at (6, -12.4) {$v_j$};
\node[draw=none, fill=none] at (7, -12.4) {$v_k$};
\node[draw=none, fill=none] at (8, -12.4) {$v_l$};

\node[draw=none, fill=none] at (6.5, -12.9) {(xx)};
\node (vi21) at (10, -12) [fill=black] {};
\node (vj21) at (11, -12) [fill=red] {};
\node (vk21) at (12, -12) [fill=green] {};
\node (vl21) at (13, -12) [fill=green] {};

\draw (vi21) .. controls (11, -11.5) .. (vk21);

\node[draw=none, fill=none] at (10, -12.4) {$v_i$};
\node[draw=none, fill=none] at (11, -12.4) {$v_j$};
\node[draw=none, fill=none] at (12, -12.4) {$v_k$};
\node[draw=none, fill=none] at (13, -12.4) {$v_l$};

\node[draw=none, fill=none] at (11.5, -12.9) {(xxi)};

\node (vi22) at (0, -14) [fill=black] {};
\node (vj22) at (1, -14) [fill=red] {};
\node (vk22) at (2, -14) [fill=green] {};
\node (vl22) at (3, -14) [fill=green] {};

\draw (vi22) .. controls (1, -13.5) .. (vk22);
\draw (vi22) .. controls (0, -14) .. (vj22);

\node[draw=none, fill=none] at (0, -14.4) {$v_i$};
\node[draw=none, fill=none] at (1, -14.4) {$v_j$};
\node[draw=none, fill=none] at (2, -14.4) {$v_k$};
\node[draw=none, fill=none] at (3, -14.4) {$v_l$};

\node[draw=none, fill=none] at (1.5, -14.9) {(xxii)};

\node (vi23) at (5, -14) [fill=black] {};
\node (vj23) at (6, -14) [fill=red] {};
\node (vk23) at (7, -14) [fill=green] {};
\node (vl23) at (8, -14) [fill=green] {};

\draw (vi23) .. controls (6, -13.5) .. (vk23);
\draw (vj23) .. controls (7, -13.5) .. (vl23);

\node[draw=none, fill=none] at (5, -14.4) {$v_i$};
\node[draw=none, fill=none] at (6, -14.4) {$v_j$};
\node[draw=none, fill=none] at (7, -14.4) {$v_k$};
\node[draw=none, fill=none] at (8, -14.4) {$v_l$};

\node[draw=none, fill=none] at (6.5, -14.9) {(xxiii)};

\node (vi24) at (10, -14) [fill=black] {};
\node (vj24) at (11, -14) [fill=red] {};
\node (vk24) at (12, -14) [fill=green] {};
\node (vl24) at (13, -14) [fill=green] {};

\draw (vi24) .. controls (11, -13.5) .. (vk24);
\draw (vj24) .. controls (12, -13.5) .. (vl24);
\draw (vi24) .. controls (10, -14) .. (vj24);

\node[draw=none, fill=none] at (10, -14.4) {$v_i$};
\node[draw=none, fill=none] at (11, -14.4) {$v_j$};
\node[draw=none, fill=none] at (12, -14.4) {$v_k$};
\node[draw=none, fill=none] at (13, -14.4) {$v_l$};

\node[draw=none, fill=none] at (11.5, -14.9) {(xxiv)};

\node (vi25) at (0, -16) [fill=green] {};
\node (vj25) at (1, -16) [fill=black] {};
\node (vk25) at (2, -16) [fill=red] {};
\node (vl25) at (3, -16) [fill=green] {};

\draw (vj25) .. controls (2, -15.5) .. (vl25);

\node[draw=none, fill=none] at (0, -16.4) {$v_i$};
\node[draw=none, fill=none] at (1, -16.4) {$v_j$};
\node[draw=none, fill=none] at (2, -16.4) {$v_k$};
\node[draw=none, fill=none] at (3, -16.4) {$v_l$};

\node[draw=none, fill=none] at (1.5, -16.9) {(xxv)};

\node (vi26) at (5, -16) [fill=green] {};
\node (vj26) at (6, -16) [fill=black] {};
\node (vk26) at (7, -16) [fill=red] {};
\node (vl26) at (8, -16) [fill=green] {};

\draw (vj26) .. controls (7, -15.5) .. (vl26);
\draw (vj26) .. controls (6, -16) .. (vk26);

\node[draw=none, fill=none] at (5, -16.4) {$v_i$};
\node[draw=none, fill=none] at (6, -16.4) {$v_j$};
\node[draw=none, fill=none] at (7, -16.4) {$v_k$};
\node[draw=none, fill=none] at (8, -16.4) {$v_l$};

\node[draw=none, fill=none] at (6.5, -16.9) {(xxvi)};

\node (vi27) at (10, -16) [fill=green] {};
\node (vj27) at (11, -16) [fill=black] {};
\node (vk27) at (12, -16) [fill=red] {};
\node (vl27) at (13, -16) [fill=green] {};

\draw (vi27) .. controls (11, -15.5) .. (vk27);
\draw (vj27) .. controls (12, -15.5) .. (vl27);

\node[draw=none, fill=none] at (10, -16.4) {$v_i$};
\node[draw=none, fill=none] at (11, -16.4) {$v_j$};
\node[draw=none, fill=none] at (12, -16.4) {$v_k$};
\node[draw=none, fill=none] at (13, -16.4) {$v_l$};

\node[draw=none, fill=none] at (11.5, -16.9) {(xxvii)};
\node (vi28) at (5, -18) [fill=green] {};
\node (vj28) at (6, -18) [fill=black] {};
\node (vk28) at (7, -18) [fill=red] {};
\node (vl28) at (8, -18) [fill=green] {};

\draw (vi28) .. controls (6, -17.5) .. (vk28);
\draw (vj28) .. controls (7, -17.5) .. (vl28);
\draw (vj28) .. controls (6, -18) .. (vk28);

\node[draw=none, fill=none] at (5, -18.4) {$v_i$};
\node[draw=none, fill=none] at (6, -18.4) {$v_j$};
\node[draw=none, fill=none] at (7, -18.4) {$v_k$};
\node[draw=none, fill=none] at (8, -18.4) {$v_l$};

\node[draw=none, fill=none] at (6.5, -18.9) {(xxviii)};

    \end{tikzpicture}
\caption{Forbidden patterns for circular-arc 3-graphs.}
\end{figure}
\begin{proof}
 Necessity: Let $B = (X_1, X_2, X_3, E)$ be a circular-arc $3$-graph with $n$ vertices. 
Then there exists a family
$\mathcal{A} = \{A_v : v \in \bigcup_{i=1}^{3} X_i\} $
of circular arcs on a host circle such that $uv \in E$ if and only if 
$A_u \cap A_v \neq \varnothing$, where $u$ and $v$ belong to different partite sets.  

Arrange the vertices of $B$ in increasing order of the clockwise endpoints of their 
corresponding circular arcs. Denote this ordering by $v_1, v_2, v_3, \ldots, v_n$. 
We now show that, under this ordering, the configurations illustrated in Figure~7 
cannot occur.  

Consider four vertices $v_i, v_j, v_k, v_\ell$ such that $i < j < k < \ell$. 
If these vertices belong to only two of the partite sets (and not all three), 
then by Theorem~4 of \cite{paul_das}, the configurations  (xvii), (xviii), (xix), and (xx) 
cannot occur.  

Now suppose the vertices $v_i, v_j, v_k, v_\ell$ belong to all three partite sets. 
For convenience, let us color the partite sets black, red, and green. 
Assume further that $v_i v_k \in E$ (i.e $A_{v_i}\cap A_{v_k}\neq\emptyset$), where $v_i\in X_\alpha$ is black and $v_k\in X_\beta$ is green.  it leads to the following two possible cases:\\
 \textbf{Case 1.}
   \begin{figure}[H]
       \centering
       \begin{tikzpicture}[line cap=round,line join=round,x=1.0cm,y=1.0cm,scale=.9]
\clip(6.,3.) rectangle (13.,9.);
\draw [line width=.5pt] (9.,6.) circle (2.cm);
\draw [shift={(9.,6.)},line width=.5pt]  plot[domain=-2.7367008673047097:0.3876695783739986,variable=\t]({1.*2.5893628559937287*cos(\t r)+0.*2.5893628559937287*sin(\t r)},{0.*2.5893628559937287*cos(\t r)+1.*2.5893628559937287*sin(\t r)});
\draw [shift={(9.,6.)},line width=.5pt]  plot[domain=0.32175055439664224:0.8204714939306736,variable=\t]({1.*3.1622776601683795*cos(\t r)+0.*3.1622776601683795*sin(\t r)},{0.*3.1622776601683795*cos(\t r)+1.*3.1622776601683795*sin(\t r)});
\begin{scriptsize}
\draw [fill=white] (9.5,7.94) circle (1.pt);
\draw (9.3,7.9) node[anchor=north west,scale=1] {$v_1$};

\draw [fill=white] (9.,8.) circle (1.pt);
\draw (8.8,8) node[anchor=north west,scale=1] {$v_n$};

\draw [fill=white] (8.5,7.94) circle (1.pt);
\draw (8.,7.9) node[anchor=north west,scale=1] {$v_{n-1}$};

\draw [fill=black] (10.83,6.8) circle (2.pt);
\draw (10.1,7.2) node[anchor=north west,scale=1.5] {$v_i$};

\draw [fill=white] (10.005240914749061,4.270985626631615) circle (2.pt);
\draw (9.5,5) node[anchor=north west,scale=1.5] {$v_j$};

\draw [fill=green] (7.17,5.2) circle (2.pt);
\draw (7.2,5.6) node[anchor=north west,scale=1.5] {$v_k$};

\draw [fill=white] (7.283275503697917,7.026088204686302) circle (2.pt);
\draw (7.2,7.3) node[anchor=north west,scale=1.5] {$v_l$};

\draw [fill=black] (11.32,7.1) circle (.5pt);
\draw [fill=black] (11.1,7.36) circle (.5pt);
\draw [fill=black] (10.86,7.62) circle (.5pt);
\draw [fill=black] (10.3,8.8) circle (.5pt);
\draw [fill=black] (10.6,8.7) circle (.5pt);
\draw [fill=black] (10.91,8.5) circle (.5pt);
\draw (8,4.2) node[anchor=north west,scale=1.5] {$A_{v_k}$};
\draw (11.8,8) node[anchor=north west,scale=1.5] {$A_{v_i}$};

\end{scriptsize}
\end{tikzpicture}
       \caption{Clockwise end point of $A_{v_i}$ lies in $A_{v_k}$ ($v_j\in X \setminus X_\beta$, $v_l\in X\setminus X_\alpha$).}
       \label{fig:enter-label-6}
   \end{figure}
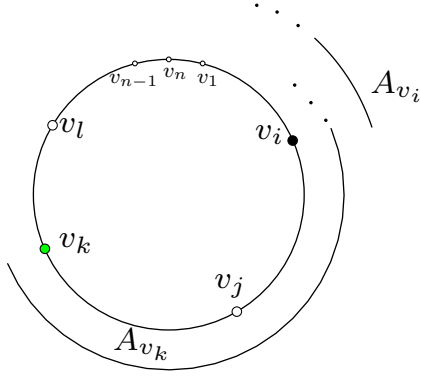
   In this case the configurations (i)-(iv), (ix)-(xii), and (xxi)-(xxiv) in Figure 7 will not occur.

   \textbf{Case 2.}
   
\begin{figure}[H]
    \centering

   \begin{tikzpicture}[line cap=round,line join=round,x=1.0cm,y=1.0cm,scale=.9]
\clip(5.,1.) rectangle (11.,8.);
\draw [line width=.5pt] (8.,5.) circle (2.cm);
\draw [shift={(8.,5.)},line width=.5pt]  plot[domain=0.3966104021074092:3.53219969728748,variable=\t]({1.*2.536927275268253*cos(\t r)+0.*2.536927275268253*sin(\t r)},{0.*2.536927275268253*cos(\t r)+1.*2.536927275268253*sin(\t r)});
\draw [shift={(8.,5.)},line width=.5pt]  plot[domain=3.4972285378905528:4.71238898038469,variable=\t]({1.*2.986904752415115*cos(\t r)+0.*2.986904752415115*sin(\t r)},{0.*2.986904752415115*cos(\t r)+1.*2.986904752415115*sin(\t r)});
\begin{scriptsize}
\draw [fill=white] (8.,7.) circle (1pt);
\draw [fill=white] (8.827605888602369,6.820732954925209) circle (1pt);
\draw [fill=white] (7.189651391201308,6.828478912161151) circle (1pt);
\draw [fill=black] (9.85,5.8) circle (2.5pt);
\draw [fill=white] (9.714985851425087,3.9710084891449458) circle (2.5pt);
\draw [fill=green] (6.13,4.3) circle (2.5pt);
\draw [fill=white] (6.269596358729952,6.002847564826958) circle (2.5pt);
\draw [fill=black] (5.74,3.92) circle (.5pt);
\draw [fill=black] (5.98,3.68) circle (.5pt);
\draw [fill=black] (6.2,3.46) circle (.5pt);
\draw [fill=black] (8.26,2.08) circle (.5pt);
\draw [fill=black] (8.64,2.16) circle (.5pt);
\draw [fill=black] (8.96,2.28) circle (.5pt);
\draw (7.8,7) node[anchor=north west,scale=1] {$v_n$};
\draw (7.,6.8) node[anchor=north west,scale=1] {$v_{n-1}$};
\draw (8.4,6.8) node[anchor=north west,scale=1] {$v_1$};
\draw (5,3) node[anchor=north west,scale=1.5] {$A_{v_k}$};
\draw (8,8.1) node[anchor=north west,scale=1.5] {$A_{v_i}$};
\draw (9,6.1) node[anchor=north west,scale=1.5] {$v_i$};
\draw (9,4.3) node[anchor=north west,scale=1.5] {$v_j$};
\draw (6.2,4.4) node[anchor=north west,scale=1.5] {$v_k$};
\draw (6.2,6.1) node[anchor=north west,scale=1.5] {$v_l$};

\end{scriptsize}
\end{tikzpicture}
 \caption{Clockwise end point of $A_{v_k}$ lies in $A_{v_i}$ ($v_j\in X \setminus X_\beta$, $v_l\in X\setminus X_\alpha$).}
    \label{fig:enter-label-7}
\end{figure}
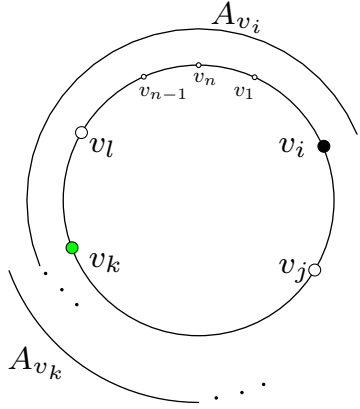
Similar to Case 1, in this case also the configurations (i)-(iv), (ix)-(xii), and (xxi)-(xxiv) in Figure 7 will not occur.\\
\par If $v_j v_l \in E$ (i.e $A_{v_j}\cap A_{v_l}\neq\emptyset$), where $v_j\in X_\alpha$ is black and $v_l\in X_\beta$ is green.  it leads to the following two possible cases:\\

   \textbf{Case 3.}
   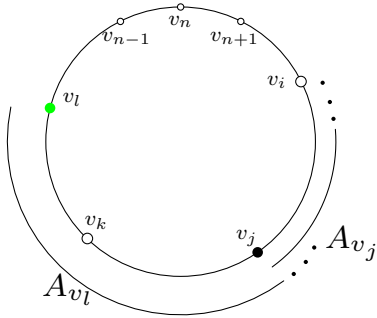
\begin{figure}[H]
       \centering
       \begin{tikzpicture}[line cap=round,line join=round,x=1.0cm,y=1.0cm,scale=.8]
\clip(2,1) rectangle (10,8);
\draw(6,5) circle (2.24cm);
\draw [shift={(6,5)}] plot[domain=2.94:5.35,variable=\t]({1*2.88*cos(\t r)+0*2.88*sin(\t r)},{0*2.88*cos(\t r)+1*2.88*sin(\t r)});
\draw [shift={(6,5)}] plot[domain=-0.94:0.08,variable=\t]({1*2.58*cos(\t r)+0*2.58*sin(\t r)},{0*2.58*cos(\t r)+1*2.58*sin(\t r)});
\begin{scriptsize}
\draw [fill=white] (8,6) circle (2.5pt);
\fill [color=black] (7.28,3.16) circle (2.5pt);
\draw[color=black] (7.1,3.4) node {$v_j$};
\draw[color=black] (7.6,6.) node {$v_i$};
\draw [fill=white] (4.45,3.39) circle (2.5pt);
\draw[color=black] (4.6,3.66) node {$v_k$};
\fill [color=green] (3.83,5.56) circle (2.5pt);
\draw[color=black] (4.2,5.7) node {$v_l$};
\fill [color=black] (7.88,2.8) circle (1.pt);
\fill [color=black] (8.08,3) circle (1.pt);
\fill [color=black] (8.22,3.24) circle (1.pt);
\fill [color=black] (8.54,5.38) circle (1.pt);
\fill [color=black] (8.48,5.66) circle (1.pt);
\fill [color=black] (8.36,5.98) circle (1pt);
\draw [fill=white] (6,7.24) circle (1.5pt);
\draw[color=black] (6,7.) node {$v_n$};

\draw[fill=white] (5,7) circle (1.5pt);
\draw[color=black] (5.1,6.7) node {$v_{n-1}$};

\draw [fill=white] (7,7) circle (1.5pt);
\draw[color=black] (6.9,6.7) node {$v_{n+1}$};
\draw[color=black] (3.5,3.) node[anchor=north west,scale=1.5] {$A_{v_l}$};
\draw[color=black] (8.2,3.8) node[anchor=north west,scale=1.5] {$A_{v_j}$};

\end{scriptsize}
\end{tikzpicture}
       \caption{Clockwise end point of $A_{v_j}$ lies in $A_{v_l}$ ($v_i\in X \setminus X_\alpha$, $v_k\in X\setminus X_\beta$).}
       \label{fig:enter-label-6}
   \end{figure}
In this case the configurations (v)-(viii), (xiii)-(xvi), and (xxv)-(xxviii) in Figure 7 will not occur.\\

   \textbf{Case 4.}
   
\begin{figure}[H]
    \centering

  \begin{tikzpicture}[line cap=round,line join=round,x=1.0cm,y=1.0cm,scale=.8]
\clip(2,1) rectangle (10,9);
\draw(6,5) circle (2.24cm);
\draw [shift={(6,5)}] plot[domain=2.91:3.6,variable=\t]({1*2.69*cos(\t r)+0*2.69*sin(\t r)},{0*2.69*cos(\t r)+1*2.69*sin(\t r)});
\draw [shift={(6,5)}] plot[domain=-0.94:2.93,variable=\t]({1*2.99*cos(\t r)+0*2.99*sin(\t r)},{0*2.99*cos(\t r)+1*2.99*sin(\t r)});
\begin{scriptsize}
\fill [color=red] (8,6) circle (2.5pt);
\fill [color=black] (7.28,3.16) circle (2.5pt);
\draw[color=black] (7.1,3.4) node {$v_j$};
\draw[color=black] (7.6,6.) node {$v_i$};

\fill [color=red] (4.45,3.39) circle (2.5pt);
\draw[color=black] (4.6,3.66) node {$v_k$};
\fill [color=green] (3.83,5.56) circle (2.5pt);
\draw[color=black] (4.2,5.7) node {$v_l$};
\fill [color=black] (3.06,5.5) circle (1.pt);
\fill [color=black] (3.07,5.3) circle (1.pt);
\fill [color=black] (3.1,5.1) circle (1.pt);
\fill [color=black] (3.8,3.5) circle (1.pt);
\fill [color=black] (3.96,3.3) circle (1.pt);
\fill [color=black] (4.138,3.1) circle (1pt);
\draw [fill=white] (6,7.24) circle (1.5pt);
\draw[color=black] (6,7.) node {$v_n$};

\draw[fill=white] (5,7) circle (1.5pt);
\draw[color=black] (5.1,6.7) node {$v_{n-1}$};

\draw [fill=white] (7,7) circle (1.5pt);
\draw[color=black] (6.9,6.7) node {$v_{n+1}$};
\draw[color=black] (2.2,5.) node[anchor=north west,scale=1.5] {$A_{v_l}$};
\draw[color=black] (8.5,6.8) node[anchor=north west,scale=1.5]{$A_{v_j}$};
\end{scriptsize}
\end{tikzpicture}
 \caption{Clockwise end point of $A_{v_l}$ lies in $A_{v_j}$.}
    \label{fig:enter-label-7}
\end{figure}
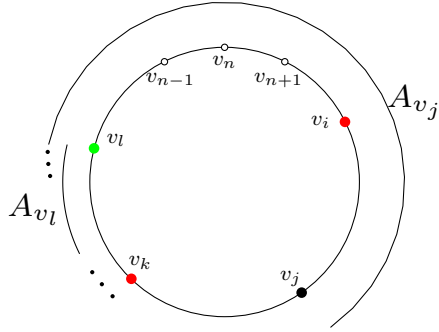
Similar to Case 3, in this case the configurations (v)-(viii), (xiii)-(xvi), and (xxv)-(xxviii) in Figure 7 will not occur.\\
Therefore, if $B=(X_1,X_2,X_3,E)$ is a circular-arc 3-graph then there exist an ordering  $v_1,v_2,...,v_n$ of the vertices of $B$ such that no indices $i<j<k<l$  occur in the configurations in Figure 7.
\par \textbf{Sufficiency:}  
Let us assume that the vertices of $B = (X_1, X_2, X_3, E)$
can be ordered as $v_1, v_2, \dots, v_n$
such that no four indices \( i < j < k < \ell \) correspond to any of the forbidden configurations shown in Figure~7.  

We now construct a family of circular arcs
$\mathcal{A} = \{ A_{v_i} : 1 \leq i \leq n \}$
associated with the vertices of \( B \).  

Suppose \( v_i \in X_\alpha \) for some \( \alpha \in \{1,2,3\} \). Define $A_{v_i} = [m_i, i], \quad 1 \leq i \leq n$,
where \( v_{m_i} \in X \setminus X_\alpha \) is the last consecutive vertex (outside the partite set \( X_\alpha \)) that is adjacent to \( v_i \) when traversing anticlockwise starting from \( v_i \).  

It remains to show that
$A_{v_i} \cap A_{v_k} \neq \varnothing 
\quad \Longleftrightarrow \quad 
v_i v_k \in E$,
where \( v_i \) and \( v_k \) belong to different partite sets.  

If \( A_{v_i} \cap A_{v_k} \neq \varnothing \), then the intersection arises in one of the two possible ways illustrated in Figures~12(i) and 12(ii). Suppose \( v_i \in X_\alpha \) and \( v_k \in X_\beta \) with \( \alpha \neq \beta \). Without loss of generality, let us assign colors: vertices of \( X_\alpha \) are colored black, vertices of \( X_\beta \) green, and the remaining vertices (outside \( X_\alpha \cup X_\beta \)) red.

 \begin{figure}[H]
     \centering
     \begin{tikzpicture}[line cap=round,line join=round,x=1.0cm,y=1.0cm,scale=.9]
\clip(-4.5,.7) rectangle (12,9.);
\draw [line width=.5pt] (7.980034032898469,4.820306296086217) circle (2.1797851460003312cm);
\draw [shift={(7.980034032898469,4.820306296086217)},line width=.5pt]  plot[domain=0.4567461636507309:3.468531325972799,variable=\t]({1.*2.6295111433068192*cos(\t r)+0.*2.6295111433068192*sin(\t r)},{0.*2.6295111433068192*cos(\t r)+1.*2.6295111433068192*sin(\t r)});
\draw [line width=.5pt] (-1.,5.) circle (2.cm);
\draw [shift={(8.,5.)},line width=.5pt]  plot[domain=3.2095853215937065:4.71238898038469,variable=\t]({1.*2.9694753582657367*cos(\t r)+0.*2.9694753582657367*sin(\t r)},{0.*2.9694753582657367*cos(\t r)+1.*2.9694753582657367*sin(\t r)});
\draw [shift={(-1.,5.)},line width=.5pt]  plot[domain=-3.0722065965203265:0.9708947512270215,variable=\t]({1.*2.33397037997512*cos(\t r)+0.*2.33397037997512*sin(\t r)},{0.*2.33397037997512*cos(\t r)+1.*2.33397037997512*sin(\t r)});
\draw [shift={(-1.,5.)},line width=.5pt]  plot[domain=0.499684617825641:1.3763509501398137,variable=\t]({1.*2.579162802238808*cos(\t r)+0.*2.579162802238808*sin(\t r)},{0.*2.579162802238808*cos(\t r)+1.*2.579162802238808*sin(\t r)});
\begin{scriptsize}
\draw [fill=white] (7.980034032898469,7.000091442086548) circle (1pt);
\draw (7.6,7) node[anchor=north west,scale=1.5] {$v_n$};

\draw [fill=black] (9.93,5.8) circle (2.5pt);
\draw [fill=green] (5.800248886898137,4.820306296086217) circle (2.5pt);
\draw [fill=black] (8.26,2.08) circle (.5pt);
\draw [fill=black] (8.64,2.16) circle (.5pt);
\draw [fill=black] (8.96,2.28) circle (.5pt);
\draw [fill=white] (-1.,7.) circle (1pt);
\draw (-1.3,7) node[anchor=north west,scale=1.5] {$v_n$};

\draw [fill=white] (5.966228857665243,3.986008987611857) circle (1.5pt);
\draw [fill=green] (-2.989990136070066,4.800151911833412) circle (2.5pt);
\draw [fill=black] (0.712607201786847,6.032945580554864) circle (2.5pt);
\draw [fill=white] (0.12912247494201323,6.650782371053442) circle (1.5pt);
\draw [fill=black] (-0.7,7.56) circle (.5pt);
\draw [fill=black] (-1.,7.54) circle (.5pt);
\draw [fill=black] (-1.3,7.5) circle (.5pt);
\draw (9.2,6.1) node[anchor=north west,scale=1.5] {$v_i$};
\draw (5.9,5.1) node[anchor=north west,scale=1.5] {$v_k$};
\draw (5.9,4.4) node[anchor=north west,scale=1.5] {$v_{m_i}$};
\draw (.1,6.1) node[anchor=north west,scale=1.5] {$v_i$};
\draw (-.5,6.7) node[anchor=north west,scale=1.5] {$v_{m_k}$};
\draw (-3,5.1) node[anchor=north west,scale=1.5] {$v_k$};
\draw (6.2,7.8) node[anchor=north west,scale=1.5] {$A_{v_i}$};
\draw (4.9,3.1) node[anchor=north west,scale=1.5] {$A_{v_k}$};
\draw (0.2,7.8) node[anchor=north west,scale=1.5] {$A_{v_i}$};
\draw (-2.2,2.8) node[anchor=north west,scale=1.5] {$A_{v_k}$};
\draw (-3,1.6) node[anchor=north west,scale=1.2] {Figure $12(i)$:};
\draw (-4.6,1.3) node[anchor=north west,scale=1.2] { clockwise end point of $A_{v_i}$ lies in $A_{v_k}$.};
\draw (6,1.6) node[anchor=north west,scale=1.2] {Figure $12(ii)$:};
\draw (4.2,1.3) node[anchor=north west,scale=1.2] { clockwise end point of $A_{v_k}$ lies in $A_{v_i}$.};

\end{scriptsize}
\end{tikzpicture}
     \label{fig:enter-label-8}
 \end{figure}
Therefore, in either case, by the construction of \(A_{v_i}\) and \(A_{v_k}\), it is clear that $v_i v_k \in E$.
Thus, \(A_{v_i} \cap A_{v_k} \neq \emptyset \) implies that \( v_i v_k \in E \).  

Now, suppose, for the sake of contradiction, that \( A_{v_i} \cap A_{v_k} = \emptyset \).  
Then, by the construction of \( A_{v_i} \) and \( A_{v_k} \), there must exist a vertex  
$v_j \notin X_\beta$, $(i < j < k)$,
such that $v_j$ is not adjacent to $v_k$.  
Additionally, there must exist another vertex  
$v_\ell \notin X_\alpha$, 
positioned between \( v_k \) and \( v_i \) in the clockwise order, which is not adjacent to \( v_i \).  

Depending on the position of \( v_\ell \), and also on the partite sets to which the vertices \( v_j \) and \( v_\ell \) belong, we obtain the following cases.

 \begin{figure}[H]
    \centering
    \begin{tikzpicture}[line cap=round,line join=round,x=1.0cm,y=1.0cm,scale=.9]
\clip(2,.5) rectangle (18.,8.);
\draw [line width=.5pt] (6.,5.) circle (2.cm);
\draw [shift={(6.,5.)},line width=.5pt]  plot[domain=-0.008333140440135445:2.14717154738608,variable=\t]({1.*2.4000833318866244*cos(\t r)+0.*2.4000833318866244*sin(\t r)},{0.*2.4000833318866244*cos(\t r)+1.*2.4000833318866244*sin(\t r)});
\draw [shift={(6.,5.)},line width=.5pt]  plot[domain=3.141592653589793:5.255715449211019,variable=\t]({1.*2.48*cos(\t r)+0.*2.48*sin(\t r)},{0.*2.48*cos(\t r)+1.*2.48*sin(\t r)});
\draw [line width=.5pt] (13.,5.) circle (2.cm);
\draw [shift={(13.,5.)},line width=.5pt]  plot[domain=3.141592653589793:5.217058339757389,variable=\t]({1.*2.48*cos(\t r)+0.*2.48*sin(\t r)},{0.*2.48*cos(\t r)+1.*2.48*sin(\t r)});
\draw [shift={(13.,5.)},line width=.5pt]  plot[domain=0.:0.8258385310050181,variable=\t]({1.*2.48*cos(\t r)+0.*2.48*sin(\t r)},{0.*2.48*cos(\t r)+1.*2.48*sin(\t r)});
\begin{scriptsize}
\draw [fill=white] (6.,7.) circle (1.5pt);
\draw (5.8,6.8) node[anchor=north west,scale=1] {$v_n$};
\draw [fill=black] (8.,5.) circle (2.5pt);
\draw (7.3,5.2) node[anchor=north west,scale=1.5] {$v_i$};

\draw [fill=red] (7.723868430315539,3.9859597468732124) circle (2.5pt);
\draw (4.1,5.3) node[anchor=north west,scale=1.5] {$v_k$};
\draw (4.3,6.2) node[anchor=north west,scale=1.5] {$v_l$};
\draw (4.9,6.8) node[anchor=north west,scale=1.5] {$v_{m_i}$};
\draw (4.1,3) node[anchor=north west,scale=1.5] {$A_{v_k}$};
\draw (6.1,7.9) node[anchor=north west,scale=1.5] {$A_{v_i}$};

\draw (7.,4.4) node[anchor=north west,scale=1.5] {$v_j$};

\draw [fill=white] (7.0201530344336085,3.279741941935484) circle (1.5pt);
\draw (6.2,3.9) node[anchor=north west,scale=1.5] {$v_{m_k}$};
\draw [fill=green] (4.,5.) circle (2.5pt);

\draw [fill=red] (4.257384112780107,5.981473315790515) circle (2.5pt);
\draw [fill=white] (4.971008489144947,6.714985851425089) circle (1.5pt);
\draw [fill=white] (13.,7.) circle (1.5pt);
\draw [fill=black] (15.,5.) circle (2.5pt);
\draw [fill=white] (14.329997221961367,6.493689187741227) circle (1.5pt);
\draw [fill=red] (13.742781352708207,6.856953381770519) circle (2.5pt);
\draw [fill=green] (11.,5.) circle (2.5pt);
\draw [fill=white] (13.972971344926611,3.2526228907440444) circle (1.5pt);
\draw [fill=red] (14.734027887970035,4.003432248293083) circle (2.5pt);
\draw (14.3,5.3) node[anchor=north west,scale=1.5] {$v_i$};
\draw (11,5.3) node[anchor=north west,scale=1.5] {$v_k$};
\draw (13.9,4.4) node[anchor=north west,scale=1.5] {$v_j$};
\draw (13.2,3.9) node[anchor=north west,scale=1.5] {$v_{m_k}$};
\draw (13.3,7.47) node[anchor=north west,scale=1.5] {$v_l$};
\draw (4.1,5.3) node[anchor=north west,scale=1.5] {$v_k$};
\draw (13.5,6.6) node[anchor=north west,scale=1.5] {$v_{m_i}$};
\draw (12.7,7.5) node[anchor=north west,scale=1.] {$v_n$};
\draw (10.8,3.) node[anchor=north west,scale=1.5] {$A_{v_k}$};
\draw (15.25,6) node[anchor=north west,scale=1.5] {$A_{v_i}$};
\draw (4.5,2.2) node[anchor=north west,scale=1.2] {Figure $13(i)$: };
\draw(2.5,1.8)node[anchor=north west,scale=1.2]{$v_i\in X_{\alpha}$, $v_k\in X_\beta$, and $v_j$, $v_l$ both in $X_\gamma$};
\draw (5.5,1.4) node[anchor=north west,scale=1.2] { $(k<l\leq n)$};

\draw (12.5,2.2) node[anchor=north west,scale=1.2] {Figure $13(ii)$:}; 
\draw(10.2,1.8)node[anchor=north west,scale=1.2]{$v_i\in X_{\alpha}$, $v_k\in X_\beta$, and $v_j$, $v_l$ both in $X_\gamma$};

\draw (13.5,1.4) node[anchor=north west,scale=1.2] {$(1\leq l<i)$.};

\end{scriptsize}
\end{tikzpicture}
    \label{fig: enter-label-9}
\end{figure}

\begin{figure}[H]
    \centering
    \begin{tikzpicture}[line cap=round,line join=round,x=1.0cm,y=1.0cm,scale=.9]
\clip(2,.5) rectangle (18.,8.);
\draw [line width=.5pt] (6.,5.) circle (2.cm);
\draw [shift={(6.,5.)},line width=.5pt]  plot[domain=-0.008333140440135445:2.14717154738608,variable=\t]({1.*2.4000833318866244*cos(\t r)+0.*2.4000833318866244*sin(\t r)},{0.*2.4000833318866244*cos(\t r)+1.*2.4000833318866244*sin(\t r)});
\draw [shift={(6.,5.)},line width=.5pt]  plot[domain=3.141592653589793:5.255715449211019,variable=\t]({1.*2.48*cos(\t r)+0.*2.48*sin(\t r)},{0.*2.48*cos(\t r)+1.*2.48*sin(\t r)});
\draw [line width=.5pt] (13.,5.) circle (2.cm);
\draw [shift={(13.,5.)},line width=.5pt]  plot[domain=3.141592653589793:5.217058339757389,variable=\t]({1.*2.48*cos(\t r)+0.*2.48*sin(\t r)},{0.*2.48*cos(\t r)+1.*2.48*sin(\t r)});
\draw [shift={(13.,5.)},line width=.5pt]  plot[domain=0.:0.8258385310050181,variable=\t]({1.*2.48*cos(\t r)+0.*2.48*sin(\t r)},{0.*2.48*cos(\t r)+1.*2.48*sin(\t r)});
\begin{scriptsize}
\draw [fill=white] (6.,7.) circle (1.5pt);
\draw (5.8,6.8) node[anchor=north west,scale=1] {$v_n$};
\draw [fill=black] (8.,5.) circle (2.5pt);
\draw (7.3,5.2) node[anchor=north west,scale=1.5] {$v_i$};

\draw [fill=black] (7.723868430315539,3.9859597468732124) circle (2.5pt);
\draw (4.1,5.3) node[anchor=north west,scale=1.5] {$v_k$};
\draw (4.3,6.2) node[anchor=north west,scale=1.5] {$v_l$};
\draw (4.9,6.8) node[anchor=north west,scale=1.5] {$v_{m_i}$};
\draw (4.1,3) node[anchor=north west,scale=1.5] {$A_{v_k}$};
\draw (6.1,7.9) node[anchor=north west,scale=1.5] {$A_{v_i}$};

\draw (7.,4.4) node[anchor=north west,scale=1.5] {$v_j$};

\draw [fill=white] (7.0201530344336085,3.279741941935484) circle (1.5pt);
\draw (6.2,3.9) node[anchor=north west,scale=1.5] {$v_{m_k}$};
\draw [fill=green] (4.,5.) circle (2.5pt);

\draw [fill=red] (4.257384112780107,5.981473315790515) circle (2.5pt);
\draw [fill=white] (4.971008489144947,6.714985851425089) circle (1.5pt);
\draw [fill=white] (13.,7.) circle (1.5pt);
\draw [fill=black] (15.,5.) circle (2.5pt);
\draw [fill=white] (14.329997221961367,6.493689187741227) circle (1.5pt);
\draw [fill=red] (13.742781352708207,6.856953381770519) circle (2.5pt);
\draw [fill=green] (11.,5.) circle (2.5pt);
\draw [fill=white] (13.972971344926611,3.2526228907440444) circle (1.5pt);
\draw [fill=black] (14.734027887970035,4.003432248293083) circle (2.5pt);
\draw (14.3,5.3) node[anchor=north west,scale=1.5] {$v_i$};
\draw (11,5.3) node[anchor=north west,scale=1.5] {$v_k$};
\draw (13.9,4.4) node[anchor=north west,scale=1.5] {$v_j$};
\draw (13.2,3.9) node[anchor=north west,scale=1.5] {$v_{m_k}$};
\draw (13.3,7.47) node[anchor=north west,scale=1.5] {$v_l$};
\draw (4.1,5.3) node[anchor=north west,scale=1.5] {$v_k$};
\draw (13.5,6.6) node[anchor=north west,scale=1.5] {$v_{m_i}$};
\draw (12.7,7.5) node[anchor=north west,scale=1.] {$v_n$};
\draw (10.8,3.) node[anchor=north west,scale=1.5] {$A_{v_k}$};
\draw (15.25,6) node[anchor=north west,scale=1.5] {$A_{v_i}$};
\draw (4.5,2.2) node[anchor=north west,scale=1.2] {Figure $13(iii)$: };
\draw(2.5,1.8)node[anchor=north west,scale=1.2]{$v_i\in X_{\alpha}$, $v_k\in X_\beta$, and $v_j\in X_{\alpha}$, $v_l\in X_\gamma$};
\draw (5.5,1.4) node[anchor=north west,scale=1.2] { $(k<l\leq n)$};

\draw (12.5,2.2) node[anchor=north west,scale=1.2] {Figure $13(iv)$:}; 
\draw(10.2,1.8)node[anchor=north west,scale=1.2]{$v_i\in X_{\alpha}$, $v_k\in X_\beta$, and $v_j\in X_{\alpha}$, $v_l\in X_\gamma$};

\draw (13.5,1.4) node[anchor=north west,scale=1.2] {$(1\leq l<i)$.};

\end{scriptsize}
\end{tikzpicture}
    \label{fig: enter-label-9}
\end{figure}

\begin{figure}[H]
    \centering
    \begin{tikzpicture}[line cap=round,line join=round,x=1.0cm,y=1.0cm,scale=.9]
\clip(2,.5) rectangle (18.,8.);
\draw [line width=.5pt] (6.,5.) circle (2.cm);
\draw [shift={(6.,5.)},line width=.5pt]  plot[domain=-0.008333140440135445:2.14717154738608,variable=\t]({1.*2.4000833318866244*cos(\t r)+0.*2.4000833318866244*sin(\t r)},{0.*2.4000833318866244*cos(\t r)+1.*2.4000833318866244*sin(\t r)});
\draw [shift={(6.,5.)},line width=.5pt]  plot[domain=3.141592653589793:5.255715449211019,variable=\t]({1.*2.48*cos(\t r)+0.*2.48*sin(\t r)},{0.*2.48*cos(\t r)+1.*2.48*sin(\t r)});
\draw [line width=.5pt] (13.,5.) circle (2.cm);
\draw [shift={(13.,5.)},line width=.5pt]  plot[domain=3.141592653589793:5.217058339757389,variable=\t]({1.*2.48*cos(\t r)+0.*2.48*sin(\t r)},{0.*2.48*cos(\t r)+1.*2.48*sin(\t r)});
\draw [shift={(13.,5.)},line width=.5pt]  plot[domain=0.:0.8258385310050181,variable=\t]({1.*2.48*cos(\t r)+0.*2.48*sin(\t r)},{0.*2.48*cos(\t r)+1.*2.48*sin(\t r)});
\begin{scriptsize}
\draw [fill=white] (6.,7.) circle (1.5pt);
\draw (5.8,6.8) node[anchor=north west,scale=1] {$v_n$};
\draw [fill=black] (8.,5.) circle (2.5pt);
\draw (7.3,5.2) node[anchor=north west,scale=1.5] {$v_i$};

\draw [fill=red] (7.723868430315539,3.9859597468732124) circle (2.5pt);
\draw (4.1,5.3) node[anchor=north west,scale=1.5] {$v_k$};
\draw (4.3,6.2) node[anchor=north west,scale=1.5] {$v_l$};
\draw (4.9,6.8) node[anchor=north west,scale=1.5] {$v_{m_i}$};
\draw (4.1,3) node[anchor=north west,scale=1.5] {$A_{v_k}$};
\draw (6.1,7.9) node[anchor=north west,scale=1.5] {$A_{v_i}$};

\draw (7.,4.4) node[anchor=north west,scale=1.5] {$v_j$};

\draw [fill=white] (7.0201530344336085,3.279741941935484) circle (1.5pt);
\draw (6.2,3.9) node[anchor=north west,scale=1.5] {$v_{m_k}$};
\draw [fill=green] (4.,5.) circle (2.5pt);

\draw [fill=green] (4.257384112780107,5.981473315790515) circle (2.5pt);
\draw [fill=white] (4.971008489144947,6.714985851425089) circle (1.5pt);
\draw [fill=white] (13.,7.) circle (1.5pt);
\draw [fill=black] (15.,5.) circle (2.5pt);
\draw [fill=white] (14.329997221961367,6.493689187741227) circle (1.5pt);
\draw [fill=green] (13.742781352708207,6.856953381770519) circle (2.5pt);
\draw [fill=green] (11.,5.) circle (2.5pt);
\draw [fill=white] (13.972971344926611,3.2526228907440444) circle (1.5pt);
\draw [fill=red] (14.734027887970035,4.003432248293083) circle (2.5pt);
\draw (14.3,5.3) node[anchor=north west,scale=1.5] {$v_i$};
\draw (11,5.3) node[anchor=north west,scale=1.5] {$v_k$};
\draw (13.9,4.4) node[anchor=north west,scale=1.5] {$v_j$};
\draw (13.2,3.9) node[anchor=north west,scale=1.5] {$v_{m_k}$};
\draw (13.3,7.47) node[anchor=north west,scale=1.5] {$v_l$};
\draw (4.1,5.3) node[anchor=north west,scale=1.5] {$v_k$};
\draw (13.5,6.6) node[anchor=north west,scale=1.5] {$v_{m_i}$};
\draw (12.7,7.5) node[anchor=north west,scale=1.] {$v_n$};
\draw (10.8,3.) node[anchor=north west,scale=1.5] {$A_{v_k}$};
\draw (15.25,6) node[anchor=north west,scale=1.5] {$A_{v_i}$};
\draw (4.5,2.2) node[anchor=north west,scale=1.2] {Figure $13(v)$: };
\draw(2.5,1.8)node[anchor=north west,scale=1.2]{$v_i\in X_{\alpha}$, $v_k\in X_\beta$, and $v_j\in X_{\gamma}$, $v_l\in X_\beta$};
\draw (5.5,1.4) node[anchor=north west,scale=1.2] { $(k<l\leq n)$};

\draw (12.5,2.2) node[anchor=north west,scale=1.2] {Figure $13(vi)$:}; 
\draw(10.2,1.8)node[anchor=north west,scale=1.2]{$v_i\in X_{\alpha}$, $v_k\in X_\beta$, and $v_j\in X_{\gamma}$, $v_l\in X_\beta$};

\draw (13.5,1.4) node[anchor=north west,scale=1.2] {$(1\leq l<i)$.};

\end{scriptsize}
\end{tikzpicture}
    \label{fig: enter-label-9}
\end{figure}

\begin{figure}[H]
	\centering
	\begin{tikzpicture}[line cap=round,line join=round,x=1.0cm,y=1.0cm,scale=.9]
		\clip(2,.5) rectangle (18.,8.);
		\draw [line width=.5pt] (6.,5.) circle (2.cm);
		\draw [shift={(6.,5.)},line width=.5pt]  plot[domain=-0.008333140440135445:2.14717154738608,variable=\t]({1.*2.4000833318866244*cos(\t r)+0.*2.4000833318866244*sin(\t r)},{0.*2.4000833318866244*cos(\t r)+1.*2.4000833318866244*sin(\t r)});
		\draw [shift={(6.,5.)},line width=.5pt]  plot[domain=3.141592653589793:5.255715449211019,variable=\t]({1.*2.48*cos(\t r)+0.*2.48*sin(\t r)},{0.*2.48*cos(\t r)+1.*2.48*sin(\t r)});
		\draw [line width=.5pt] (13.,5.) circle (2.cm);
		\draw [shift={(13.,5.)},line width=.5pt]  plot[domain=3.141592653589793:5.217058339757389,variable=\t]({1.*2.48*cos(\t r)+0.*2.48*sin(\t r)},{0.*2.48*cos(\t r)+1.*2.48*sin(\t r)});
		\draw [shift={(13.,5.)},line width=.5pt]  plot[domain=0.:0.8258385310050181,variable=\t]({1.*2.48*cos(\t r)+0.*2.48*sin(\t r)},{0.*2.48*cos(\t r)+1.*2.48*sin(\t r)});
		\begin{scriptsize}
			\draw [fill=white] (6.,7.) circle (1.5pt);
			\draw (5.8,6.8) node[anchor=north west,scale=1] {$v_n$};
			\draw [fill=black] (8.,5.) circle (2.5pt);
			\draw (7.3,5.2) node[anchor=north west,scale=1.5] {$v_i$};
			
			\draw [fill=black] (7.723868430315539,3.9859597468732124) circle (2.5pt);
			\draw (4.1,5.3) node[anchor=north west,scale=1.5] {$v_k$};
			\draw (4.3,6.2) node[anchor=north west,scale=1.5] {$v_l$};
			\draw (4.9,6.8) node[anchor=north west,scale=1.5] {$v_{m_i}$};
			\draw (4.1,3) node[anchor=north west,scale=1.5] {$A_{v_k}$};
			\draw (6.1,7.9) node[anchor=north west,scale=1.5] {$A_{v_i}$};
			
			\draw (7.,4.4) node[anchor=north west,scale=1.5] {$v_j$};
			
			\draw [fill=white] (7.0201530344336085,3.279741941935484) circle (1.5pt);
			\draw (6.2,3.9) node[anchor=north west,scale=1.5] {$v_{m_k}$};
			\draw [fill=green] (4.,5.) circle (2.5pt);
			
			\draw [fill=green] (4.257384112780107,5.981473315790515) circle (2.5pt);
			\draw [fill=white] (4.971008489144947,6.714985851425089) circle (1.5pt);
			\draw [fill=white] (13.,7.) circle (1.5pt);
			\draw [fill=black] (15.,5.) circle (2.5pt);
			\draw [fill=white] (14.329997221961367,6.493689187741227) circle (1.5pt);
			\draw [fill=green] (13.742781352708207,6.856953381770519) circle (2.5pt);
			\draw [fill=green] (11.,5.) circle (2.5pt);
			\draw [fill=white] (13.972971344926611,3.2526228907440444) circle (1.5pt);
			\draw [fill=black] (14.734027887970035,4.003432248293083) circle (2.5pt);
			\draw (14.3,5.3) node[anchor=north west,scale=1.5] {$v_i$};
			\draw (11,5.3) node[anchor=north west,scale=1.5] {$v_k$};
			\draw (13.9,4.4) node[anchor=north west,scale=1.5] {$v_j$};
			\draw (13.2,3.9) node[anchor=north west,scale=1.5] {$v_{m_k}$};
			\draw (13.3,7.47) node[anchor=north west,scale=1.5] {$v_l$};
			\draw (4.1,5.3) node[anchor=north west,scale=1.5] {$v_k$};
			\draw (13.5,6.6) node[anchor=north west,scale=1.5] {$v_{m_i}$};
			\draw (12.7,7.5) node[anchor=north west,scale=1.] {$v_n$};
			\draw (10.8,3.) node[anchor=north west,scale=1.5] {$A_{v_k}$};
			\draw (15.25,6) node[anchor=north west,scale=1.5] {$A_{v_i}$};
			\draw (4.5,2.2) node[anchor=north west,scale=1.2] {Figure $13(vii)$: };
			\draw(2.5,1.8)node[anchor=north west,scale=1.2]{$v_i\in X_{\alpha}$, $v_k\in X_\beta$, and $v_j\in X_{\alpha}$, $v_l\in X_\beta$};
			\draw (5.5,1.4) node[anchor=north west,scale=1.2] { $(k<l\leq n)$};
			
			\draw (12.5,2.2) node[anchor=north west,scale=1.2] {Figure $13(viii)$:}; 
			\draw(10.2,1.8)node[anchor=north west,scale=1.2]{$v_i\in X_{\alpha}$, $v_k\in X_\beta$, and $v_j\in X_{\alpha}$, $v_l\in X_\beta$};
			
			\draw (13.5,1.4) node[anchor=north west,scale=1.2] {$(1\leq l<i)$.};

		\end{scriptsize}
	\end{tikzpicture}
	\label{fig: enter-label-9}
\end{figure}

Consider the following cases:  \\

\indent	 \textbf{Case 1:} $k < l$, with $v_i \in X_\alpha$, $v_k \in X_\beta$, and $v_j, v_l \in X_\gamma$ (Figure~13(i)).  
	The vertices $v_i, v_j, v_k, v_l$ (with $i < j < k < l$) form one of the configurations (i)--(iv) of Figure~7, depending on whether $v_iv_j$ and $v_kv_l$ belong to $E$. Each possibility leads to a contradiction.\\  
	
 \textbf{Case 2:} $l < i$, with $v_i \in X_\alpha$, $v_k \in X_\beta$, and $v_j, v_l \in X_\gamma$ (Figure~13(ii)).  
	The vertices $v_l, v_i, v_j, v_k$ (with $l < i < j < k$) yield one of the configurations (v)--(viii) of Figure~7, depending on whether $v_iv_j$ and $v_kv_l$ belong to $E$. Relabeling $l, i, j, k$ as $i, j, k, l$ shows that this again reduces to one of these configurations. In all cases, a contradiction arises.\\  
	
 \textbf{Case 3:} $k < l$, with $v_i \in X_\alpha$, $v_k \in X_\beta$, $v_j \in X_\alpha$, and $v_l \in X_\gamma$ (Figure~13(iii)).  
	The vertices $v_i, v_j, v_k, v_l$ (with $i < j < k < l$) form one of the configurations (ix)--(xii), depending on whether $v_jv_l$ and $v_kv_l$ belong to $E$. Each case yields a contradiction. \\ 
	
	\textbf{Case 4:} $l < i$, with $v_i \in X_\alpha$, $v_k \in X_\beta$, $v_j \in X_\alpha$, and $v_l \in X_\gamma$ (Figure~13(iv)).  
	The vertices $v_l, v_i, v_j, v_k$ (with $l < i < j < k$) form one of the configurations (xiii)--(xvi). After relabeling the indices, the same contradiction follows.\\  
	
 \textbf{Case 5:} $k < l$, with $v_i \in X_\alpha$, $v_k \in X_\beta$, $v_j \in X_\gamma$, and $v_l \in X_\beta$ (Figure~13(v)).  
	The vertices $v_i, v_j, v_k, v_l$ (with $i < j < k < l$) form one of the configurations (xxi)--(xxiv), depending on whether $v_iv_j$ and $v_jv_l$ belong to $E$. This yields a contradiction. \\ 
	
\textbf{Case 6:} $l < i$, with $v_i \in X_\alpha$, $v_k \in X_\beta$, $v_j \in X_\gamma$, and $v_l \in X_\beta$ (Figure~13(vi)).  
	The vertices $v_l, v_i, v_j, v_k$ (with $l < i < j < k$) form one of the configurations (xxv)--(xxviii) of Figure~7, depending on whether $v_iv_j$ and $v_jv_l$ belong to $E$. After relabeling $l, i, j, k$ as $i, j, k, l$, this again reduces to one of the configurations (xxv)--(xxviii). Thus, a contradiction follows.\\  
	
\textbf{Case 7:} $k < l$, with $v_i \in X_\alpha$, $v_k \in X_\beta$, $v_j \in X_\alpha$, and $v_l \in X_\beta$ (Figure~13(vii)).  
	The vertices $v_i, v_j, v_k, v_l$ (with $i < j < k < l$) form one of the configurations (xvii)--(xviii), depending on whether $v_jv_l \in E$. In either case, a contradiction follows.\\  
	
 \textbf{Case 8:} $l < i$, with $v_i \in X_\alpha$, $v_k \in X_\beta$, $v_j \in X_\alpha$, and $v_l \in X_\beta$ (Figure~13(viii)).  
	The vertices $v_l, v_i, v_j, v_k$ (with $l < i < j < k$) form one of the configurations (xix)--(xx), depending on whether $v_jv_l \in E$. After relabeling, this again yields a contradiction.  

\noindent
In every case, we are led to a contradiction. Therefore, if $v_iv_k \in E$, it must be that
\[
A_{v_i} \cap A_{v_k} \neq \varnothing.
\]
Equivalently,
\[
v_iv_k \in E \quad \Longleftrightarrow \quad A_{v_i} \cap A_{v_k} \neq \varnothing,
\]
and thus $B = (X_1, X_2, X_3, E)$ is a circular-arc $3$-graph.

\end{proof}
\begin{theo}
Let $B = (X_1, X_2, \ldots, X_r, E)$ be an $r$-partite graph ($r \geq 4$). Then the following 
statements are equivalent:
\begin{enumerate}
    \item $B$ is a circular-arc $r$-graph;
    \item The vertices of $B$ can be ordered $v_1, v_2, \ldots, v_n$ such that no indices 
    $i < j < k < \ell$ occur in the configurations shown in Figures~7 and~14. 
    (Here, different colors indicate that vertices of different colors belong to 
    different partite sets; moreover, all edges not explicitly drawn are absent.)
\end{enumerate}
\end{theo}

\begin{figure}[H]
\centering

\begin{tikzpicture}[scale=1.2, every node/.style={circle, inner sep=2pt}]

\node (vi1) at (0, 0) [fill=green] {};
\node (vj1) at (1, 0) [fill=blue] {};
\node (vk1) at (2, 0) [fill=red] {};
\node (vl1) at (3, 0) [fill=black] {};

\draw (vi1) .. controls (1, .5) .. (vk1);

\node[draw=none, fill=none] at (0, -0.4) {$v_i$};
\node[draw=none, fill=none] at (1, -0.4) {$v_j$};
\node[draw=none, fill=none] at (2, -0.4) {$v_k$};
\node[draw=none, fill=none] at (3, -0.4) {$v_l$};

\node[draw=none, fill=none] at (1.5, -0.9) {(i)};

\node (vi2) at (5, 0) [fill=green] {};
\node (vj2) at (6, 0) [fill=blue] {};
\node (vk2) at (7, 0) [fill=red] {};
\node (vl2) at (8, 0) [fill=black] {};

\draw (vi2) .. controls (6., 0.5) .. (vk2);
\draw (vi2) .. controls (5, 0) .. (vj2);

\node[draw=none, fill=none] at (5, -0.4) {$v_i$};
\node[draw=none, fill=none] at (6, -0.4) {$v_j$};
\node[draw=none, fill=none] at (7, -0.4) {$v_k$};
\node[draw=none, fill=none] at (8, -0.4) {$v_l$};

\node[draw=none, fill=none] at (6.5, -0.9) {(ii)};

\node (vi3) at (10, 0) [fill=green] {};
\node (vj3) at (11, 0) [fill=blue] {};
\node (vk3) at (12, 0) [fill=red] {};
\node (vl3) at (13, 0) [fill=black] {};

\draw (vi3) .. controls (11, 0.5) .. (vk3);
\draw (vk3) .. controls (12, 0) .. (vl3);
\node[draw=none, fill=none] at (10, -.4) {$v_i$};
\node[draw=none, fill=none] at (11, -.4) {$v_j$};
\node[draw=none, fill=none] at (12, -.4) {$v_k$};
\node[draw=none, fill=none] at (13, -.4) {$v_l$};

\node[draw=none, fill=none] at (11.5, -0.9) {(iii)};

\node (vi4) at (0, -2) [fill=green] {};
\node (vj4) at (1, -2) [fill=blue] {};
\node (vk4) at (2, -2) [fill=red] {};
\node (vl4) at (3, -2) [fill=black] {};

\draw (vi4) .. controls (1, -1.5) .. (vk4);
\draw (vj4) .. controls (2, -1.5) .. (vl4);

\node[draw=none, fill=none] at (0, -2.4) {$v_i$};
\node[draw=none, fill=none] at (1, -2.4) {$v_j$};
\node[draw=none, fill=none] at (2, -2.4) {$v_k$};
\node[draw=none, fill=none] at (3, -2.4) {$v_l$};

\node[draw=none, fill=none] at (1.5, -2.9) {(iv)};

\node (vi5) at (5, -2) [fill=green] {};
\node (vj5) at (6, -2) [fill=blue] {};
\node (vk5) at (7, -2) [fill=red] {};
\node (vl5) at (8, -2) [fill=black] {};

\draw (vi5) .. controls (6, -1.5) .. (vk5);
\draw (vk5) .. controls (7, -2) .. (vl5);
\draw (vi5) .. controls (5, -2) .. (vj5);

\node[draw=none, fill=none] at (5, -2.4) {$v_i$};
\node[draw=none, fill=none] at (6, -2.4) {$v_j$};
\node[draw=none, fill=none] at (7, -2.4) {$v_k$};
\node[draw=none, fill=none] at (8, -2.4) {$v_l$};

\node[draw=none, fill=none] at (6.5, -2.9) {(v)};

\node (vi6) at (10, -2) [fill=green] {};
\node (vj6) at (11, -2) [fill=blue] {};
\node (vk6) at (12, -2) [fill=red] {};
\node (vl6) at (13, -2) [fill=black] {};

\draw (vi6) .. controls (11, -1.5) .. (vk6);
\draw (vj6) .. controls (12, -1.5) .. (vl6);
\draw (vi6) .. controls (11, -2) .. (vj6);

\node[draw=none, fill=none] at (10, -2.4) {$v_i$};
\node[draw=none, fill=none] at (11, -2.4) {$v_j$};
\node[draw=none, fill=none] at (12, -2.4) {$v_k$};
\node[draw=none, fill=none] at (13, -2.4) {$v_l$};

\node[draw=none, fill=none] at (11.5, -2.9) {(vi)};

\node (vi4) at (0, -4) [fill=green] {};
\node (vj4) at (1, -4) [fill=blue] {};
\node (vk4) at (2, -4) [fill=red] {};
\node (vl4) at (3, -4) [fill=black] {};

\draw (vi4) .. controls (1, -3.5) .. (vk4);
\draw (vj4) .. controls (2, -3.5) .. (vl4);
\draw (vk4) .. controls (2, -4) .. (vl4);

\node[draw=none, fill=none] at (0, -4.4) {$v_i$};
\node[draw=none, fill=none] at (1, -4.4) {$v_j$};
\node[draw=none, fill=none] at (2, -4.4) {$v_k$};
\node[draw=none, fill=none] at (3, -4.4) {$v_l$};

\node[draw=none, fill=none] at (1.5, -4.9) {(vii)};

\node (vi5) at (5, -4) [fill=green] {};
\node (vj5) at (6, -4) [fill=blue] {};
\node (vk5) at (7, -4) [fill=red] {};
\node (vl5) at (8, -4) [fill=black] {};

\draw (vi5) .. controls (6, -3.5) .. (vk5);
\draw (vj5) .. controls (7, -3.5) .. (vl5);
\draw (vi5) .. controls (5, -4) .. (vj5);
\draw (vk5) .. controls (7, -4) .. (vl5);

\node[draw=none, fill=none] at (5, -4.4) {$v_i$};
\node[draw=none, fill=none] at (6, -4.4) {$v_j$};
\node[draw=none, fill=none] at (7, -4.4) {$v_k$};
\node[draw=none, fill=none] at (8, -4.4) {$v_l$};

\node[draw=none, fill=none] at (6.5, -4.9) {(viii)};

\node (vi6) at (10, -4) [fill=green] {};
\node (vj6) at (11, -4) [fill=blue] {};
\node (vk6) at (12, -4) [fill=red] {};
\node (vl6) at (13, -4) [fill=black] {};

\draw (vj6) .. controls (12, -3.5) .. (vl6);

\node[draw=none, fill=none] at (10, -4.4) {$v_i$};
\node[draw=none, fill=none] at (11, -4.4) {$v_j$};
\node[draw=none, fill=none] at (12, -4.4) {$v_k$};
\node[draw=none, fill=none] at (13, -4.4) {$v_l$};

\node[draw=none, fill=none] at (11.5, -4.9) {(ix)};

\node (vi4) at (0, -6) [fill=green] {};
\node (vj4) at (1, -6) [fill=blue] {};
\node (vk4) at (2, -6) [fill=red] {};
\node (vl4) at (3, -6) [fill=black] {};

\draw (vj4) .. controls (1, -6) .. (vk4);
\draw (vj4) .. controls (2, -5.5) .. (vl4);

\node[draw=none, fill=none] at (0, -6.4) {$v_i$};
\node[draw=none, fill=none] at (1, -6.4) {$v_j$};
\node[draw=none, fill=none] at (2, -6.4) {$v_k$};
\node[draw=none, fill=none] at (3, -6.4) {$v_l$};

\node[draw=none, fill=none] at (1.5, -6.9) {(x)};

\node (vi5) at (5, -6) [fill=green] {};
\node (vj5) at (6, -6) [fill=blue] {};
\node (vk5) at (7, -6) [fill=red] {};
\node (vl5) at (8, -6) [fill=black] {};

\draw (vi5) .. controls (6.5, -5.5) .. (vl5);
\draw (vj5) .. controls (7, -5.3) .. (vl5);

\node[draw=none, fill=none] at (5, -6.4) {$v_i$};
\node[draw=none, fill=none] at (6, -6.4) {$v_j$};
\node[draw=none, fill=none] at (7, -6.4) {$v_k$};
\node[draw=none, fill=none] at (8, -6.4) {$v_l$};

\node[draw=none, fill=none] at (6.5, -6.9) {(xi)};
\node (vi5) at (10, -6) [fill=green] {};
\node (vj5) at (11, -6) [fill=blue] {};
\node (vk5) at (12, -6) [fill=red] {};
\node (vl5) at (13, -6) [fill=black] {};

\draw (vi5) .. controls (11.5, -5.5) .. (vl5);
\draw (vj5) .. controls (12, -5.3) .. (vl5);
\draw (vj5) .. controls (11, -6) .. (vk5);

\node[draw=none, fill=none] at (10, -6.4) {$v_i$};
\node[draw=none, fill=none] at (11, -6.4) {$v_j$};
\node[draw=none, fill=none] at (12, -6.4) {$v_k$};
\node[draw=none, fill=none] at (13, -6.4) {$v_l$};

\node[draw=none, fill=none] at (11.5, -6.9) {(xii)};

\node (vi4) at (0, -8) [fill=green] {};
\node (vj4) at (1, -8) [fill=blue] {};
\node (vk4) at (2, -8) [fill=red] {};
\node (vl4) at (3, -8) [fill=black] {};

\draw (vi4) .. controls (1, -7.5) .. (vk4);
\draw (vj4) .. controls (2, -7.5) .. (vl4);
\draw (vj4) .. controls (1, -8) .. (vk4);

\node[draw=none, fill=none] at (0, -8.4) {$v_i$};
\node[draw=none, fill=none] at (1, -8.4) {$v_j$};
\node[draw=none, fill=none] at (2, -8.4) {$v_k$};
\node[draw=none, fill=none] at (3, -8.4) {$v_l$};

\node[draw=none, fill=none] at (1.5, -8.9) {(xiii)};

\node (vi4) at (5, -8) [fill=green] {};
\node (vj4) at (6, -8) [fill=blue] {};
\node (vk4) at (7, -8) [fill=red] {};
\node (vl4) at (8, -8) [fill=black] {};

\draw (vi4) .. controls (6.5, -7.2) .. (vl4);
\draw (vj4) .. controls (7, -7.6) .. (vl4);
\draw (vi4) .. controls (6, -7.6) .. (vk4);

\node[draw=none, fill=none] at (5, -8.4) {$v_i$};
\node[draw=none, fill=none] at (6, -8.4) {$v_j$};
\node[draw=none, fill=none] at (7, -8.4) {$v_k$};
\node[draw=none, fill=none] at (8, -8.4) {$v_l$};

\node[draw=none, fill=none] at (6.5, -8.9) {(xiv)};
\node (vi4) at (10, -8) [fill=green] {};
\node (vj4) at (11, -8) [fill=blue] {};
\node (vk4) at (12, -8) [fill=red] {};
\node (vl4) at (13, -8) [fill=black] {};

\draw (vi4) .. controls (11.5, -7.2) .. (vl4);
\draw (vj4) .. controls (12, -7.6) .. (vl4);
\draw (vi4) .. controls (11, -7.6) .. (vk4);
\draw (vj4) .. controls (11, -8) .. (vk4);
\node[draw=none, fill=none] at (10, -8.4) {$v_i$};
\node[draw=none, fill=none] at (11, -8.4) {$v_j$};
\node[draw=none, fill=none] at (12, -8.4) {$v_k$};
\node[draw=none, fill=none] at (13, -8.4) {$v_l$};

\node[draw=none, fill=none] at (11.5, -8.9) {(xv)};

	\draw (2,-6.5) node[anchor=north west,scale=1.] {Figure 14: Some forbidden patterns for circular-arc $r$-graphs ($r\geq 4)$.};
\end{tikzpicture}
\end{figure}

\begin{proof}
Necessity:Let $B = (X_1, X_2, \dots, X_r, E)$ be a circular-arc $r$-graph with $n$ vertices.  
There exists a family $\mathcal{A} = \{A_v : v \in \bigcup_{i=1}^{r} X_i\}$
of circular arcs on a host circle such that $uv \in E$ if and only if 
$A_u \cap A_v \neq \varnothing$, where $u$ and $v$ belong to different partite sets.  

Arrange the vertices of $B$ in increasing order of the clockwise endpoints of their 
corresponding circular arcs. Denote this ordering by $v_1, v_2, v_3, \ldots, v_n$. 
We now show that, under this ordering, the configurations illustrated in Figure~7 
and Figure~14 cannot occur.  

Consider four vertices $v_i, v_j, v_k, v_\ell$ such that $i < j < k < \ell$.  
If these vertices belong to at most three of the partite sets (but not four distinct 
partite sets), then by Theorem~5 it follows immediately that the configurations in 
Figure~7 cannot occur.  

Now suppose that $v_i, v_j, v_k, v_\ell$ belong to four different partite sets. 
For clarity, we use four distinct colors---black, red, green, and blue---to represent 
vertices from different partite sets. Assume that $v_i v_k \in E$ 
(i.e., $A_{v_i} \cap A_{v_k} \neq \varnothing$), where $v_i$ is colored green 
and $v_k$ is colored red. This leads to the following two possible cases:
\\
\textbf{Case 1.}
\begin{figure}[H]
	\centering
	\begin{tikzpicture}[line cap=round,line join=round,x=1.0cm,y=1.0cm,scale=.9]
		\clip(4.,2.5) rectangle (15.,9.);
		\draw [line width=.5pt] (9.,6.) circle (2.cm);
		\draw [shift={(9.,6.)},line width=.5pt]  plot[domain=-2.7367008673047097:0.3876695783739986,variable=\t]({1.*2.5893628559937287*cos(\t r)+0.*2.5893628559937287*sin(\t r)},{0.*2.5893628559937287*cos(\t r)+1.*2.5893628559937287*sin(\t r)});
		\draw [shift={(9.,6.)},line width=.5pt]  plot[domain=0.32175055439664224:0.8204714939306736,variable=\t]({1.*3.1622776601683795*cos(\t r)+0.*3.1622776601683795*sin(\t r)},{0.*3.1622776601683795*cos(\t r)+1.*3.1622776601683795*sin(\t r)});
		\begin{scriptsize}
			\draw [fill=white] (9.5,7.94) circle (1.pt);
			\draw (9.3,7.9) node[anchor=north west,scale=1] {$v_1$};
			
			\draw [fill=white] (9.,8.) circle (1.pt);
			\draw (8.8,8) node[anchor=north west,scale=1] {$v_n$};
			
			\draw [fill=white] (8.5,7.94) circle (1.pt);
			\draw (8.,7.9) node[anchor=north west,scale=1] {$v_{n-1}$};
			
			\draw [fill=green] (10.83,6.8) circle (2.5pt);
			\draw (10.1,7.2) node[anchor=north west,scale=1.5] {$v_i$};
			
			\draw [fill=blue] (10.005240914749061,4.270985626631615) circle (2.5pt);
			\draw (9.5,5) node[anchor=north west,scale=1.5] {$v_j$};
			
			\draw [fill=red] (7.17,5.2) circle (2.5pt);
			\draw (7.2,5.6) node[anchor=north west,scale=1.5] {$v_k$};
			
			\draw [fill=black] (7.283275503697917,7.026088204686302) circle (2.5pt);
			\draw (7.2,7.3) node[anchor=north west,scale=1.5] {$v_l$};
			
			\draw [fill=black] (11.32,7.1) circle (.5pt);
			\draw [fill=black] (11.1,7.36) circle (.5pt);
			\draw [fill=black] (10.86,7.62) circle (.5pt);
			\draw [fill=black] (10.3,8.8) circle (.5pt);
			\draw [fill=black] (10.6,8.7) circle (.5pt);
			\draw [fill=black] (10.91,8.5) circle (.5pt);
			\draw (8,4.2) node[anchor=north west,scale=1.5] {$A_{v_k}$};
			\draw (11.8,8) node[anchor=north west,scale=1.5] {$A_{v_i}$};
			
			\draw (4,3.) node[anchor=north west,scale=1.2] {Figure 15: Clockwise end point of $A_{v_i}$ lies in $A_{v_k}$.};
		\end{scriptsize}
	\end{tikzpicture}
	
\end{figure}
In this case, irrespective of whether $v_kv_l, v_iv_j, v_jv_l \in E$ or $v_kv_l, v_iv_j, v_jv_l \notin E$, the configurations (i)--(viii) of Figure~14 are avoided.
\\
   
\textbf{Case 2.}
\begin{figure}[H]
	\centering

	\begin{tikzpicture}[line cap=round,line join=round,x=1.0cm,y=1.0cm,scale=.9]
		\clip(3.,1.) rectangle (13.,8.);
		\draw [line width=.5pt] (8.,5.) circle (2.cm);
		\draw [shift={(8.,5.)},line width=.5pt]  plot[domain=0.3966104021074092:3.53219969728748,variable=\t]({1.*2.536927275268253*cos(\t r)+0.*2.536927275268253*sin(\t r)},{0.*2.536927275268253*cos(\t r)+1.*2.536927275268253*sin(\t r)});
		\draw [shift={(8.,5.)},line width=.5pt]  plot[domain=3.4972285378905528:4.71238898038469,variable=\t]({1.*2.986904752415115*cos(\t r)+0.*2.986904752415115*sin(\t r)},{0.*2.986904752415115*cos(\t r)+1.*2.986904752415115*sin(\t r)});
		\begin{scriptsize}
			\draw [fill=white] (8.,7.) circle (1pt);
			\draw [fill=white] (8.827605888602369,6.820732954925209) circle (1pt);
			\draw [fill=white] (7.189651391201308,6.828478912161151) circle (1pt);
			\draw [fill=green] (9.85,5.8) circle (2.5pt);
			\draw [fill=blue] (9.714985851425087,3.9710084891449458) circle (2.5pt);
			\draw [fill=red] (6.13,4.3) circle (2.5pt);
			\draw [fill=black] (6.269596358729952,6.002847564826958) circle (2.5pt);
			\draw [fill=black] (5.74,3.92) circle (.5pt);
			\draw [fill=black] (5.98,3.68) circle (.5pt);
			\draw [fill=black] (6.2,3.46) circle (.5pt);
			\draw [fill=black] (8.26,2.08) circle (.5pt);
			\draw [fill=black] (8.64,2.16) circle (.5pt);
			\draw [fill=black] (8.96,2.28) circle (.5pt);
			\draw (7.8,7) node[anchor=north west,scale=1] {$v_n$};
			\draw (7.,6.8) node[anchor=north west,scale=1] {$v_{n-1}$};
			\draw (8.4,6.8) node[anchor=north west,scale=1] {$v_1$};
			\draw (5,3) node[anchor=north west,scale=1.5] {$A_{v_k}$};
			\draw (8,8.1) node[anchor=north west,scale=1.5] {$A_{v_i}$};
			\draw (9,6.1) node[anchor=north west,scale=1.5] {$v_i$};
			\draw (9,4.3) node[anchor=north west,scale=1.5] {$v_j$};
			\draw (6.2,4.4) node[anchor=north west,scale=1.5] {$v_k$};
			\draw (6.2,6.1) node[anchor=north west,scale=1.5] {$v_l$};
				\draw (3,1.5) node[anchor=north west,scale=1.2] {Figure 16:Clockwise end point of $A_{v_k}$ lies in $A_{v_i}$.};
		\end{scriptsize}
	\end{tikzpicture}
\end{figure}  
Also, in this case, depending on whether the edges $v_kv_l$, $v_iv_j$, and $v_jv_l$ belong to $E$ or not, none of the configurations (i)--(viii) in Figure~14 can occur.\\[6pt]
Now, if $v_jv_l \in E$ (that is, $A_{v_j} \cap A_{v_l} \neq \varnothing$), then we obtain the following cases:
\\
 \textbf{Case 3.}
\begin{figure}[H]
	\centering
	\begin{tikzpicture}[line cap=round,line join=round,x=1.0cm,y=1.0cm,scale=.8]
		\clip(0,1) rectangle (11,8);
		\draw(6,5) circle (2.24cm);
		\draw [shift={(6,5)}] plot[domain=2.94:5.35,variable=\t]({1*2.88*cos(\t r)+0*2.88*sin(\t r)},{0*2.88*cos(\t r)+1*2.88*sin(\t r)});
		\draw [shift={(6,5)}] plot[domain=-0.94:0.08,variable=\t]({1*2.58*cos(\t r)+0*2.58*sin(\t r)},{0*2.58*cos(\t r)+1*2.58*sin(\t r)});
		\begin{scriptsize}
			\draw [fill=green] (8,6) circle (2.5pt);
			\fill [color=blue] (7.28,3.16) circle (2.5pt);
			\draw[color=black] (7.1,3.4) node {$v_j$};
			\draw[color=black] (7.6,6.) node {$v_i$};
			\draw [fill=red] (4.45,3.39) circle (2.5pt);
			\draw[color=black] (4.6,3.66) node {$v_k$};
			\fill [color=black] (3.83,5.56) circle (2.5pt);
			\draw[color=black] (4.2,5.7) node {$v_l$};
			\fill [color=black] (7.88,2.8) circle (1.pt);
			\fill [color=black] (8.08,3) circle (1.pt);
			\fill [color=black] (8.22,3.24) circle (1.pt);
			\fill [color=black] (8.54,5.38) circle (1.pt);
			\fill [color=black] (8.48,5.66) circle (1.pt);
			\fill [color=black] (8.36,5.98) circle (1pt);
			\draw [fill=white] (6,7.24) circle (1.5pt);
			\draw[color=black] (6,7.) node {$v_n$};
			
			\draw[fill=white] (5,7) circle (1.5pt);
			\draw[color=black] (5.1,6.7) node {$v_{n-1}$};
			
			\draw [fill=white] (7,7) circle (1.5pt);
			\draw[color=black] (6.9,6.7) node {$v_{n+1}$};
			\draw[color=black] (3.5,3.) node[anchor=north west,scale=1.5] {$A_{v_l}$};
			\draw[color=black] (8.2,3.8) node[anchor=north west,scale=1.5] {$A_{v_j}$};
			\draw (0,1.6) node[anchor=north west,scale=1.2] {Figure 17: Clockwise end point of $A_{v_j}$ lies in $A_{v_l}$.};
		\end{scriptsize}
	\end{tikzpicture}
\end{figure}
In this case, depending on whether the edges $v_jv_k$, $v_iv_l$, and $v_iv_k$ belong to $E$ or not, none of the configurations (iv) and (ix)--(xv) in Figure~14 can occur.
\\

\textbf{Case 4.}

\begin{figure}[H]
	\centering

	\begin{tikzpicture}[line cap=round,line join=round,x=1.0cm,y=1.0cm,scale=.8]
		\clip(0,1) rectangle (12,9);
		\draw(6,5) circle (2.24cm);
		\draw [shift={(6,5)}] plot[domain=2.91:3.6,variable=\t]({1*2.69*cos(\t r)+0*2.69*sin(\t r)},{0*2.69*cos(\t r)+1*2.69*sin(\t r)});
		\draw [shift={(6,5)}] plot[domain=-0.94:2.93,variable=\t]({1*2.99*cos(\t r)+0*2.99*sin(\t r)},{0*2.99*cos(\t r)+1*2.99*sin(\t r)});
		\begin{scriptsize}
			\fill [color=green] (8,6) circle (2.5pt);
			\fill [color=blue] (7.28,3.16) circle (2.5pt);
			\draw[color=black] (7.1,3.4) node {$v_j$};
			\draw[color=black] (7.6,6.) node {$v_i$};
			
			\fill [color=red] (4.45,3.39) circle (2.5pt);
			\draw[color=black] (4.6,3.66) node {$v_k$};
			\fill [color=black] (3.83,5.56) circle (2.5pt);
			\draw[color=black] (4.2,5.7) node {$v_l$};
			\fill [color=black] (3.06,5.5) circle (1.pt);
			\fill [color=black] (3.07,5.3) circle (1.pt);
			\fill [color=black] (3.1,5.1) circle (1.pt);
			\fill [color=black] (3.8,3.5) circle (1.pt);
			\fill [color=black] (3.96,3.3) circle (1.pt);
			\fill [color=black] (4.138,3.1) circle (1pt);
			\draw [fill=white] (6,7.24) circle (1.5pt);
			\draw[color=black] (6,7.) node {$v_n$};
			
			\draw[fill=white] (5,7) circle (1.5pt);
			\draw[color=black] (5.1,6.7) node {$v_{n-1}$};
			
			\draw [fill=white] (7,7) circle (1.5pt);
			\draw[color=black] (6.9,6.7) node {$v_{n+1}$};
			\draw[color=black] (2.1,5.) node[anchor=north west,scale=1.5] {$A_{v_l}$};
			\draw[color=black] (8.5,6.8) node[anchor=north west,scale=1.5]{$A_{v_j}$};
			\draw (0,1.6) node[anchor=north west,scale=1.2] {Figure 18: Clockwise end point of $A_{v_l}$ lies in $A_{v_j}$.};
		\end{scriptsize}
	\end{tikzpicture}
\end{figure}
Similar to Case~3, in this situation as well, depending on whether the edges $v_jv_k$, $v_iv_l$, and $v_iv_k$ belong to $E$ or not, none of the configurations (iv) and (ix)--(xv) in Figure~8 can occur.\\[6pt]
Therefore, if $B = (X_1, X_2, \dots, X_r, E)$ is a circular-arc $r$-graph, then there exists an ordering $v_1, v_2, \dots, v_n$ of the vertices of $B$ such that no indices $i < j < k < l$ give rise to any of the configurations shown in Figure~7 and Figure~14.

\par \textbf{Sufficiency:}  
Let us assume that the vertices of 
$B = (X_1, X_2, \dots, X_r, E)$
can be ordered as $v_1, v_2, \dots, v_n$ such that no four indices $i < j < k < l$ correspond to any of the forbidden configurations shown in Figure~7 and Figure~14.  

We now construct a family of circular arcs
$\mathcal{A} = \{ A_{v_i} : 1 \leq i \leq n \}$,
associated with the vertices of $B$.  

Suppose $v_i \in X_\alpha$ for some $\alpha \in \{1,2,\dots,r\}$. Define
$A_{v_i} = [m_i, i], \quad 1 \leq i \leq n$,
where $v_{m_i} \in X \setminus X_\alpha$ (with $X=\bigcup_{t=1}^{r} X_t$) is the last consecutive vertex (outside the partite set $X_\alpha$) that is adjacent to $v_i$ when traversing anticlockwise starting from $v_i$.  

It remains to show that
$A_{v_i} \cap A_{v_k} \neq \varnothing 
\quad \Longleftrightarrow \quad v_i v_k \in E$,
where $v_i$ and $v_k$ belong to different partite sets.
If $A_{v_i} \cap A_{v_k} \neq \varnothing$, then the intersection arises in one of the two possible ways (see Figures~12(i) and 12(ii)). In either case, by the construction of $A_{v_i}$ and $A_{v_k}$, it follows that $v_iv_k \in E$. Hence,
$A_{v_i} \cap A_{v_k} \neq \varnothing \;\; \Longrightarrow \;\; v_iv_k \in E$.\\
Now, suppose for the sake of contradiction that $v_iv_k \in E$ but $A_{v_i} \cap A_{v_k} = \varnothing$, where $v_i \in X_\alpha$, $v_k \in X_\beta$, and $\alpha \neq \beta \in \{1,2,\dots,r\}$.  
By the construction of $A_{v_i}$ and $A_{v_k}$, there must exist a vertex $v_j \notin X_\beta$ with $i < j < k$, such that $v_j$ is not adjacent to $v_k$.  
Additionally, there must exist another vertex $v_l \notin X_\alpha$, positioned between $v_k$ and $v_i$ in the clockwise order, which is not adjacent to $v_i$.  

If the four vertices $v_i, v_j, v_k, v_l$ together belong to at most three distinct partite sets, then by an argument similar to the proof of Theorem~5, one of the configurations in Figure~7 must occur, leading to a contradiction.  

Therefore, we may assume that the vertices $v_i, v_j, v_k, v_l$ come from four different partite sets. For clarity, we represent them with four distinct colors: green for $v_i$, blue for $v_j$, red for $v_k$, and black for $v_l$. Depending on the position of $v_l$, we obtain the following two figures:

 \begin{figure}[H]
	\centering
	\begin{tikzpicture}[line cap=round,line join=round,x=1.0cm,y=1.0cm,scale=.9]
		\clip(2,.5) rectangle (18.,8.);
		\draw [line width=.5pt] (6.,5.) circle (2.cm);
		\draw [shift={(6.,5.)},line width=.5pt]  plot[domain=-0.008333140440135445:2.14717154738608,variable=\t]({1.*2.4000833318866244*cos(\t r)+0.*2.4000833318866244*sin(\t r)},{0.*2.4000833318866244*cos(\t r)+1.*2.4000833318866244*sin(\t r)});
		\draw [shift={(6.,5.)},line width=.5pt]  plot[domain=3.141592653589793:5.255715449211019,variable=\t]({1.*2.48*cos(\t r)+0.*2.48*sin(\t r)},{0.*2.48*cos(\t r)+1.*2.48*sin(\t r)});
		\draw [line width=.5pt] (13.,5.) circle (2.cm);
		\draw [shift={(13.,5.)},line width=.5pt]  plot[domain=3.141592653589793:5.217058339757389,variable=\t]({1.*2.48*cos(\t r)+0.*2.48*sin(\t r)},{0.*2.48*cos(\t r)+1.*2.48*sin(\t r)});
		\draw [shift={(13.,5.)},line width=.5pt]  plot[domain=0.:0.8258385310050181,variable=\t]({1.*2.48*cos(\t r)+0.*2.48*sin(\t r)},{0.*2.48*cos(\t r)+1.*2.48*sin(\t r)});
		\begin{scriptsize}
			\draw [fill=white] (6.,7.) circle (1.5pt);
			\draw (5.8,6.8) node[anchor=north west,scale=1] {$v_n$};
			\draw [fill=green] (8.,5.) circle (2.5pt);
			\draw (7.3,5.2) node[anchor=north west,scale=1.5] {$v_i$};
			
			\draw [fill=blue] (7.723868430315539,3.9859597468732124) circle (2.5pt);
			\draw (4.1,5.3) node[anchor=north west,scale=1.5] {$v_k$};
			\draw (4.3,6.2) node[anchor=north west,scale=1.5] {$v_l$};
			\draw (4.9,6.8) node[anchor=north west,scale=1.5] {$v_{m_i}$};
			\draw (4.1,3) node[anchor=north west,scale=1.5] {$A_{v_k}$};
			\draw (6.1,7.9) node[anchor=north west,scale=1.5] {$A_{v_i}$};
			
			\draw (7.,4.4) node[anchor=north west,scale=1.5] {$v_j$};
			
			\draw [fill=white] (7.0201530344336085,3.279741941935484) circle (1.5pt);
			\draw (6.2,3.9) node[anchor=north west,scale=1.5] {$v_{m_k}$};
			\draw [fill=red] (4.,5.) circle (2.5pt);
			
			\draw [fill=black] (4.257384112780107,5.981473315790515) circle (2.5pt);
			\draw [fill=white] (4.971008489144947,6.714985851425089) circle (1.5pt);
			\draw [fill=white] (13.,7.) circle (1.5pt);
			\draw [fill=green] (15.,5.) circle (2.5pt);
			\draw [fill=white] (14.329997221961367,6.493689187741227) circle (1.5pt);
			\draw [fill=black] (13.742781352708207,6.856953381770519) circle (2.5pt);
			\draw [fill=red] (11.,5.) circle (2.5pt);
			\draw [fill=white] (13.972971344926611,3.2526228907440444) circle (1.5pt);
			\draw [fill=blue] (14.734027887970035,4.003432248293083) circle (2.5pt);
			\draw (14.3,5.3) node[anchor=north west,scale=1.5] {$v_i$};
			\draw (11,5.3) node[anchor=north west,scale=1.5] {$v_k$};
			\draw (13.9,4.4) node[anchor=north west,scale=1.5] {$v_j$};
			\draw (13.2,3.9) node[anchor=north west,scale=1.5] {$v_{m_k}$};
			\draw (13.3,7.47) node[anchor=north west,scale=1.5] {$v_l$};
			\draw (4.1,5.3) node[anchor=north west,scale=1.5] {$v_k$};
			\draw (13.5,6.6) node[anchor=north west,scale=1.5] {$v_{m_i}$};
			\draw (12.7,7.5) node[anchor=north west,scale=1.] {$v_n$};
			\draw (10.8,3.) node[anchor=north west,scale=1.5] {$A_{v_k}$};
			\draw (15.25,6) node[anchor=north west,scale=1.5] {$A_{v_i}$};
			\draw (3.5,2.2) node[anchor=north west,scale=1.2] {Figure $19(i)$: $(k<l\leq n).$ };
			
			\draw (11.,2.2) node[anchor=north west,scale=1.2] {Figure $19(ii)$: $(1\leq l<i)$.}; 
			

		\end{scriptsize}
	\end{tikzpicture}
	\label{fig: enter-label-9}
\end{figure}
From Figure~19(i), depending on whether the edges 
$v_iv_j, v_kv_l, v_jv_l$ belong to $E$ or not, we obtain one of the 
configurations (i)--(viii) from Figure~14, which contradicts our assumption. 
Similarly, from Figure~19(ii), depending on the presence or absence of the 
edges $v_iv_j, v_kv_l, v_jv_l$ in $E$, we obtain one of the configurations 
(iv) or (ix)--(xv) in Figure~14 (after relabeling $l,i,j,k$ as $i,j,k,l$).
\noindent
In every case, we are led to a contradiction. Therefore, if $v_iv_k \in E$, it must be that
\[
A_{v_i} \cap A_{v_k} \neq \varnothing.
\]
Equivalently,
\[
v_iv_k \in E \quad \Longleftrightarrow \quad A_{v_i} \cap A_{v_k} \neq \varnothing,
\]
and thus $B = (X_1, X_2,\dots, X_r, E)$ is a circular-arc $r$-graph. 
\end{proof}

\section{Conclusion}
The recognition algorithm for circular-arc graphs was established in linear time after extensive research \cite{dhh,McConnel}. More recently, Francis, Hell, and Stacho \cite{fhs} developed a certifying recognition algorithm for circular-arc graphs with running time $\mathcal{O}(n^3)$, based on forbidden structures of circular-arc graphs. However, the problem of designing efficient recognition algorithms for circular-arc $r$-graphs ($r\geq 2$) remains open. We hope that this paper serves as a motivating step toward resolving this problem.
\\


\end{document}